\documentclass{LMCS} 
\usepackage{amssymb}
\usepackage{amsmath}
\usepackage{latexsym}
\usepackage{stmaryrd} 
\usepackage{amstext}
\usepackage{xspace}
\usepackage{enumerate}
\usepackage{extarrows}
\usepackage{bbb}
\usepackage{hyperref}

\def\fps@figure{tp}      % Top, or separate page.
\def\fps@table{tp}
\newcommand{\pfn}[1]{\mathsf{#1}}  %font for programs
\newcommand{\cfn}[1]{\mathsf{#1}}  %font for programs
\newcommand{\ownfnt}[1]{{\mathsf{#1}}}  %font for owners

\newcommand{\DPI}{\ensuremath{\pfn{Dpi}}\xspace}
\newcommand{\picost}{\ensuremath{\pfn{Picost}}\xspace}

\newcommand{\with}{\mathbin \rhd}
\newcommand{\Own}{\ensuremath{\pfn{Own}}\xspace}
\newcommand{\Chan}{\ensuremath{\pfn{Chan}}\xspace}
\newcommand{\News}{\ensuremath{\pfn{News}}\xspace}
\newcommand{\Var}{\ensuremath{\mathsf{Var}}\xspace}

\newcommand{\cancom}[3]{({\ownfnt{#1}},{#2},{\ownfnt{#3}}) \xspace}
\newcommand{\cancomin}[2]{({\ownfnt{#1}},{#2}) \xspace}
\newcommand{\cancomout}[2]{({#1},{\ownfnt{#2}}) \xspace}
\newcommand{\ext}[3]{\mathsf{ext}({\ownfnt{#1}},{#2},{\ownfnt{#3}}) \xspace}

\newcommand{\Act}{\ensuremath{\mathsf{Act}}\xspace}
\newcommand{\record}{{\scriptstyle \mathsf{rec}}}
\newcommand{\Kost}{\ensuremath{K}\xspace}

\newcommand{\cost}[1]{\mathop{\mathsf{weight}(#1)}}
\newcommand{\costN}{\mathsf{weight}}

\newcommand{\nats}{\ensuremath{\bbb N}\xspace}

\newcommand{\ownO}{ {\ownfnt O}}
\newcommand{\sobs}{{\ownfnt e}}

\newcommand{\amort}[1]{\mathrel{\sqsubseteq_{\rm wgt}^{#1}}}
\newcommand{\behav}[1]{\mathrel{\sqsubseteq_{\rm behav}^{#1}}}

\newcommand{\cxtequiv}[1]{\mathrel{\sqsubseteq_{\rm cxt}^{#1}}}
\newcommand{\Ecxtequiv}[1]{\mathrel{\sqsubseteq_{\rm \sobs cxt}^{#1}}}
\newcommand{\Obscxtequiv}[2]{\mathrel{\sqsubseteq_{#1:\rm cxt}^{#2}}}
\newcommand{\Ocxtequiv}[1]{\Obscxtequiv{\ownO}{#1}}
\newcommand{\Gammaobs}{\Gamma^{\sobs}}
\newcommand{\Deltaobs}{\Delta^{\sobs}}

\newcommand{\Eamort}[1]{\mathrel{\sqsubseteq_{\rm \sobs wgt}^{#1}}}
\newcommand{\aamort}[1]{\mathrel{\sqsubseteq_{\rm awgt}^{#1}}}
\newcommand{\Oamort}[1]{\mathrel{\sqsubseteq_{\rm Owgt}^{#1}}}
\newcommand{\Obsaamort}[2]{\mathrel{\sqsubseteq_{\rm #1 awgt}^{#2}}}
\newcommand{\Oaamort}[1]{\Obsaamort{\ownO}{ #1}}
\newcommand{\Gammadyn}{\Gamma_{\scriptstyle dyn}}

\def\pair(#1,#2){\langle #1 , #2 \rangle}% pairs
\newcommand{\parrow}{ \mathbin{\rightharpoonup}}
\newcommand{\dom}[1]{\mathop{\text{dom}}(#1)}

\newcommand{\setof}[2]{\{ \, #1 \, \mid \, #2 \, \}}% set comprehension
\newcommand{\sset}[1]{\{ {#1}  \}  } % singleton set

\newcommand{\typeletter}[1]{{\mathsf{#1}}}
\newcommand{\tR}{\typeletter{R}}
\newcommand{\tK}{\typeletter{K}}

\newcommand{\tC}{\typeletter{C}}
\newcommand{\tE}{\typeletter{E}}

\newcommand{\calR}{\mathbin{\mathcal{R}}}

\newcommand{\calC}{\mathcal{C}}
\newcommand{\calD}{\mathcal{D}}

\makeatletter

\usepackage{ifthen}
\newcommand{\pa}[1]{\!\left(#1\right)}
\newcommand{\pc}[1]{\langle#1\rangle}
\newcommand{\Cassoc}[2]{ {{#1}\kern -0.2em : \kern -0.2em {#2}}}
\newcommand{\Cnew}[2]{(\mathop{\pfn{new}} \Cassoc{#1}{#2})}
\newcommand{\CnewNT}[1]{(\pfn{new} {#1})}
\newcommand{\Cnil}{\mathop{\pfn{0}}}
\newcommand{\Cstop}{\mathop{\pfn{stop}}}
\newcommand{\Cpar}{\mathbin{|}}

\newcommand{\Cmatch}[3]{\pfn{if}\; #1 \; \pfn{then}\; #2\;
                        \pfn{else}\; #3}
\newcommand{\Crec}[2]{\pfn{rec}\;#1.\;#2}
\newcommand{\Cloc}[2]{[#1]_{\ownfnt{#2}}}
\DeclareSymbolFont{smallcaps}{\encodingdefault}{\rmdefault}{m}{sc}
\DeclareSymbolFontAlphabet\mathsc{smallcaps}
\newcommand{\Rulef}[1]{{\mathsc{#1}}}   %%font for rules
\newcommand{\Rred}[1]{\ensuremath{\Rulef{(r\textrm{-}#1)}}}%%%%reduction
\newcommand{\Rstruct}[1]{\ensuremath{\Rulef{(s\textrm{-}#1)}}}%%%% structural equivalence
\newcommand{\Rlts}[1]{\ensuremath{\Rulef{(l\textrm{-}#1)}}}%%%%
\newcommand{\mnowidth}[2][l]{\makebox[0cm][#1]{$#2$}}%
\RequirePackage{ifthen}   % provides \ifthenelse
\newcommand{\ifemptyelse}[3]{\ifx\@@bullshit#1\@@bullshit#2\else#3\fi}
\newcommand{\ifnotempty}[2]{\ifx\@@bullshit#1\@@bullshit\else#2\fi}
\newenvironment{scope}
  {\bgroup\ignorespaces}
  {\egroup\global\@ignoretrue}

\newbox\@topbox
\newbox\@botbox
\newcommand{\linfer}[3][{}]{%
  \ifemptyelse{#1}{
    \setbox\@topbox\hbox{%
      \renewcommand{\arraystretch}{1}%
      $\begin{array}{l}#2\end{array}$}%
    }{
    \ifemptyelse{#2}{
      \setbox\@topbox\hbox{%
        \renewcommand{\arraystretch}{1}%
        $\begin{array}{l}\mnowidth[l]{\scriptstyle#1}\\\ \end{array}$}%
      }{
      \setbox\@topbox\hbox{%
        \renewcommand{\arraystretch}{1}%
        $\begin{array}{l}\mnowidth[l]{\scriptstyle#1}\\#2\end{array}$}%
      }
    }
  \setbox\@botbox\hbox{%
    \def\tilde##1{\widetilde{##1}}%
    \renewcommand{\arraystretch}{1}%
    $\begin{array}{l}{#3}\end{array}$}%
  \ifdim\wd\@botbox<\wd\@topbox% which is bigger?  set the other to match
    \setbox\@botbox\hbox to \wd\@topbox{\box\@botbox\hfil}
  \else
    \setbox\@topbox\hbox to \wd\@botbox{\box\@topbox\hfil}
  \fi
  \frac{\box\@topbox}{\box\@botbox}
  }
\newcommand{\linferSIDE}[4][{}]{%
  \ifemptyelse{#1}{
    \setbox\@topbox\hbox{%
      \renewcommand{\arraystretch}{1}%
      $\begin{array}{l}#2\end{array}$}%
    }{
    \ifemptyelse{#2}{
      \setbox\@topbox\hbox{%
        \renewcommand{\arraystretch}{1}%
        $\begin{array}{l}\mnowidth[l]{\scriptstyle#1}\\\ \end{array}$}%
      }{
      \setbox\@topbox\hbox{%
        \renewcommand{\arraystretch}{1}%
        $\begin{array}{l}\mnowidth[l]{\scriptstyle#1}\\#2\end{array}$}%
      }
    }
  \setbox\@botbox\hbox{%
    \def\tilde##1{\widetilde{##1}}%
    \renewcommand{\arraystretch}{1}%
    $\begin{array}{l}#3\end{array}$}%
  \ifdim\wd\@botbox<\wd\@topbox% which is bigger?  set the other to match
    \setbox\@botbox\hbox to \wd\@topbox{\box\@botbox\hfil}
  \else
    \setbox\@topbox\hbox to \wd\@botbox{\box\@topbox\hfil}
  \fi
  \frac{\box\@topbox}{\box\@botbox}
  \hbox{\;\footnotesize$
    \begin{scope}%
      \renewcommand{\arraystretch}{1}%
      \begin{array}[c]{l}%
        #4%
      \end{array}%
    \end{scope}%
    $}%
}

\newcommand{\slinfer}[2][{}]{\renewcommand{\arraystretch}{1}%
                  \begin{array}{l}{\scriptstyle #1}\\#2\end{array}}
\newcommand{\slinferSIDE}[3][{}]{\renewcommand{\arraystretch}{1}%
                  \begin{array}{ll}{\scriptstyle #1}\\#2 &#3\end{array}}

\def\mhto{\@transition\mharrowfill}
\def\goesto{\@transition\rightarrowfill}
\def\Goesto{\@transition\Rightarrowfill}
\def\ngoesto{\@transition\nrightarrowfill}
\def\nsgoesto{\@transition\nmapstofill}
\def\comesfrom{\@transition\leftarrowfill}
\def\nGoesto{\@transition\nRightarrowfill}
\def\@transition#1{\@ifnextchar[{\@@transition{#1}}{\@@transition{#1}[]}}
\newbox\@transbox
\newbox\@arrowbox
\def\mharrowfill{$\m@th\mathord{+}\mkern-6mu%
  \cleaders\hbox{$\mkern-2mu\mathord-\mkern-2mu$}\hfill
  \mkern-6mu\mathord\longmapsto$}
\def\mapstofill{$\m@th\mathord\mapstochar\mathord-\mkern-6mu%
  \cleaders\hbox{$\mkern-2mu\mathord-\mkern-2mu$}\hfill
  \mkern-6mu\mathord\rightarrow$}
\def\Rightarrowfill{$\m@th\mathord=\mkern-6mu%
  \cleaders\hbox{$\mkern-2mu\mathord=\mkern-2mu$}\hfill
  \mkern-6mu\mathord\Rightarrow$}
\def\nrightarrowfill{$\m@th\mathord-\mkern-6mu%
  \cleaders\hbox{$\mkern-2mu\mathord-\mkern-2mu$}\hfill
  \mkern-6mu\mathord\not\mathord-\mkern-6mu
  \cleaders\hbox{$\mkern-2mu\mathord-\mkern-2mu$}\hfill
  \mkern-6mu\mathord\rightarrow$}
\def\@@transition#1[#2]%
   {\setbox\@transbox\hbox
       {\vrule height 1.2ex depth 1.1ex width          %% changed by JK 
         0ex\hskip0.25em$\scriptstyle#2$\hskip0.25em}  
   \ifdim\wd\@transbox<1.5em
      \setbox\@transbox\hbox to 1.5em{\hfil\box\@transbox\hfil}\fi
   \setbox\@arrowbox\hbox to \wd\@transbox{#1}
   \ht\@arrowbox\z@\dp\@arrowbox\z@
   \setbox\@transbox\hbox{$\mathop{\box\@arrowbox}\limits^{\box\@transbox}$}
   \ht\@transbox\z@\dp\@transbox\z@
   \mathrel{\box\@transbox}}
\def\nxmapsto{\@transition\nmapstofill}
\def\xmapsto{\@transition\mapstofill}
\def\mapstofill{$\m@th\mathord\mapstochar\mathord-\mkern-6mu%
  \cleaders\hbox{$\mkern-2mu\mathord-\mkern-2mu$}\hfill
  \mkern-6mu\mathord\rightarrow$}
\def\nmapstofill{$\m@th\mathord\mapstochar\mathord-\mkern-6mu%
  \cleaders\hbox{$\mkern-2mu\mathord-\mkern-2mu$}\hfill
  \mkern-6mu\mathord\arrownot\mathord-\mkern-6mu
  \cleaders\hbox{$\mkern-2mu\mathord-\mkern-2mu$}\hfill
  \mkern-6mu\mathord\rightarrow$}
\newcommand{\ar}[1]{\mathbin{\xlongrightarrow{ #1}}}  %comes from amsmath

\newcommand{\arO}[2]{\mathrel{
      \goesto[{#1}]_{\raisebox{7pt}{$\scriptscriptstyle #2$}}^{\raisebox{-4pt}{$\scriptscriptstyle \ownO$}}}}

\newcommand{\arGen}[3]{\mathrel{
      \goesto[{#1}]_{\raisebox{7pt}{$\scriptscriptstyle #2$}}^{\raisebox{-4pt}{$\scriptscriptstyle   #3 $}}}}

\newcommand{\darO}[2]{\mathrel{
      \Goesto[{#1}]_{\raisebox{7pt}{$\scriptscriptstyle #2$}}^{\raisebox{-4pt}{$\scriptscriptstyle \ownO$}}}}

\newcommand{\newar}[1]{\mathrel{\overset{ #1}\mapsto}}   % abstract arrows
\newcommand{\newarstar}[1]{\mathrel{\overset{#1}\mapsto^{\raisebox{-5pt}{$\scriptstyle *$}}}}  % 

\newcommand{\dar}[1]{\xLongrightarrow{\vspace{-2em}#1}}

\newcommand{\smalleval}{\longrightarrow}  %reduction semantics

\newcommand{\structeq}{\equiv}
\newcommand{\lsetbar}{\mathopen{\{\hspace{-.24em}\vert}}
\newcommand{\rsetbar}{\mathclose{\vert\hspace{-.24em}\}}}
\newcommand{\smallscript}{\mathchoice{\scriptstyle}{\scriptstyle}{\scriptscriptstyle}{\scriptscriptstyle}}
\newcommand{\@subst@brackets}[1]{\lsetbar #1 \rsetbar}
\newcommand{\@subst}[2]{{\overset{#2}{\phantom{.}}} \!/\! {\smallscript #1}}
\newcommand{\subst}[2]{\@subst@brackets{\@subst{#1}{#2}}}
\newcommand{\fn}[1]{\mathop{\pfn{fn}}(#1)}

\newcommand{\names}[1]{\mathop{\pfn{n}}(#1)}

\newcommand{\leaveout}[1]{ }

\newcommand{\EndDefBox}{\null\hfill$\blacksquare$}
\global\let\EndProof\EndProofBox
\newcommand{\boxHere}{\global\let\EndProof\empty\EndDefBox}

\newenvironment{simple-definition}{\begin{defii}\rm}{\end{defii}}
\makeatother

\usepackage{url}

\def\doi{7 (1:7) 2011}
\lmcsheading%
{\doi}
{1--35}
{}
{}
{Jul.~30, 2009}
{Mar.~23, 2011}
{}

\begin{document}

\title[costed computations]{A calculus for costed computations}

\author[M.~Hennessey]{Matthew Hennessy}

\address{Department of Computer Science\\ Trinity College Dublin\\
    Ireland} 

\email{matthew.hennessy@cs.tcd.ie}   
\thanks{The financial support of SFI is gratefully acknowledged.}

%\date{\today}

% %  \input{abstract.body}

\begin{abstract}

We develop a version of the picalculus \picost where  
channels are interpreted as \emph{resources} which have 
costs associated with them. Code runs under the financial
responsibility of \emph{owners}; they   must pay to 
\emph{use} resources, but may profit by \emph{providing} 
them.

We provide a  proof methodology for processes described in \picost
based on bisimulations. The underlying  behavioural theory is justified
via a contextual characterisation. We also demonstrate its usefulness
via examples. 

\end{abstract}

\subjclass{F.3.1, F.3.2, F.3.3}
\keywords{resources, cost, picalculus, bisimulations, amortisation}

\maketitle

\section{Introduction}

The purpose of this paper is to develop a  behavioural
theory of processes, in which computations depend on the ability to
fund the resources involved. The theory will be based on the
well-known concept of bisimulations, \cite{pimilner}, which
automatically gives a powerful co-inductive proof methodology for
establishing properties of processes; here these properties will
include the cost of behaviour.

 We take as a starting point
the well-known picalculus, \cite{pibook,pimilner}, a language for
describing mobile processes which has a well-developed behavioural
theory. In the picalculus a process is described in terms of its
ability to input and output on \emph{communication channels}. Here we
interpret these channels as \emph{resources}, or services, as for
example in \cite{beppe}. So input along a channel, written as
$c?\pa{x}.P$ in the picalculus, is now interpreted as \emph{providing}
the service $c$, while output, written $c!\pc{v}.P$, is interpreted as
a request to \emph{use} the service $c$.  A process is now determined
by the manner in which it \emph{provides}  services and
\emph{uses} them.
   
Viewed from this perspective, we extend the picalculus in two ways.
Firstly we associate a \emph{cost} with resources; specifically for
each resource we assume that a certain amount of funds $k_u$ is
charged to \emph{use} it, and an amount $k_p$ is also required to
\emph{provide} it. Secondly we introduce \emph{principals} or \emph{owners}
who provide the funds necessary for the functioning of resources. 
The novel construct in the language is $\Cloc{P}{o}$, representing
the (picalculus) process $P$ running under the financial responsibility 
of $\ownfnt o$. For example in $\Cloc{c!\pc{v}.Q}{o}$ the use of the resource 
$c$ is only possible if $\ownfnt o$ can fund the charges. Similarly with
 $\Cloc{c?\pa{x}.Q}{o}$, but here there is also the potential for gain
for owner $\ownfnt o$; in our formulation $\ownfnt o$ profits from 
any difference between the cost in providing the resource and the 
charge made to use it. 

Our language \picost is presented in Section~\ref{sec:lang}, and is essentially 
a variation on \DPI, a typed distributed version of the picalculus, \cite{dpibook}.
The reduction semantics is given in terms of judgements of the form
\begin{align*}
  (\Gamma \with M) \smalleval (\Delta \with N)
\end{align*}
where $\Gamma,\Delta$ are \emph{cost environments}. These have a static component, giving
the costs associated with resources, and a dynamic part, which gives the funds available to owners 
and also records expenditure. The usefulness of the language is demonstrated by a series of 
simple examples. 

But the main achievement of the paper is a  behavioural theory, expressed as judgements
\begin{align}\label{eq:judge}
  (\Gamma \with M) \aamort{} (\Delta \with N) 
\end{align}
indicating that, informally speaking,
\begin{enumerate}[(i)]
\item the process $M$ running relative to the cost environment $\Gamma$ is bisimilar, in the
standard sense \cite{ccs}, with process $N$ running relative to $\Delta$

\item the costs associated with $(\Delta \with N)$ are no more, and possibly less,  than those
associated with $(\Gamma \with M)$.
\end{enumerate}
% Compositionality is reflected in the fact that this behavioural preorder will satisfy
% the property:
% \begin{align*}
%   (\Gamma \with M) \aamort{} (\Delta \with N)  \;\;\text{implies}\;\;
% (\Gamma \with M \Cpar R) \aamort{} (\Delta \with N \Cpar R) 
% \end{align*}
% where $M_1 \Cpar M_2$ represents the two processes $M_1,\;M_2$ running 
% in parallel.

 \noindent Influenced by \cite{astrid} we first develop a general framework of
\emph{weighted labelled transition systems} or \emph{wLTSs}, in which
actions, including internal actions, may have multiple weights
associated with them. We then define a notion of \emph{amortised
  weighted bisimulations} between their states, giving rise to a
preorder $s \aamort{} t$, meaning that $s, t$ are bisimilar but in some
sense the behaviours of $t$ are \emph{lighter} than those of $s$. From
this we obtain, in the standard manner, a co-inductive proof
methodology for proving that two systems are related; it is sufficient
to find, or construct, a particular amortised weighted bisimulation 
containing the pair $(s,t)$. 

This proof methodology is applied to \picost by first interpreting the
language as an LTS, in agreement with the reduction semantics, and
then interpreting this LTS as a wLTS, giving rise to (parametrised 
versions of) the judgements
(\ref{eq:judge}) above.  But as we will see these judgements can be
interpreted in two ways. If the recorded expenditure represents
\emph{costs} then $(\Delta \with N) $ can be considered an improvement
on $ (\Gamma \with M)$. On the other hand if it represents
\emph{profits} then we have the reverse; $ (\Gamma \with M)$ is an
improvement on $(\Delta \with N)$ as it has the potential to be
heavier.

The details of this theory are given in Section~\ref{sec:lts}, and the
resulting proof methodology is illustrated by examples.  However in
Section~\ref{sec:cxt} we re-examine this proof methodology, in the
light of reasonable properties we would expect of it; and these are
found wanting.  It turns out that the manner in which we generate the
wLTS for \picost from its operational semantics is too coarse. We
show how to generate a somewhat more abstract wLTS, and prove that the
resulting proof methodology is satisfactory, in a precise technical
sense,  by adapting the notion of reduction barbed congruence,
\cite{ht92,pibook,pityping,dpibook}.

\begin{figure}[t]

\begin{align*}
&\begin{array}{lcll}
  M,\;N &::=&                &\textbf{}\\
        &   & \Cloc{T}{o}    &\text{Owned code}\\
        &   &  M \Cpar N     &\text{Composition}\\
        &   &  \Cnew{r}{\tR}M &\text{Scoped resource}\\
        &   &  \Cnil         & \text{Identity}\\\\\\\\
\end{array}
\\
%&\qquad
&\begin{array}{lcll}
T,\;U &::=&                   &\\
      &  & u?\pa{x}. T       &\text{Provide resource $u$}\\
      &  & u!\pc{v}. T               &\text{Use resource $u$}\\
      &  & \Cmatch{v=v}{T}{U} &\text{Matching}\\
      &  & \Cnew{r}{\tR} T     &\text{Resource creation}\\    
      &  & T \Cpar U          &\text{Concurrency} \\
      &  & \Crec{X}{T}                &\text{Recursion} \\
      &  & X                &\text{Recursion variable} \\
      &  & \Cstop               &\text{Termination}      \\
%     &  & \Cdel{a}{v}      &\text{Asynchronous   message delivery}\\
\end{array}
\end{align*}
\caption{Syntax of \picost\label{fig:syn}}
\end{figure}

\section{The language \picost}\label{sec:lang}

\subsection{Syntax:}

We assume a set of \emph{channel} or \emph{resource} names \Chan,
ranged over by $a,b,c,\ldots,$ $r,\ldots $ whose use requires some cost,
a distinct set of (value) variables $\Var$, ranged over by
$x,y,\ldots$, and a further distinct set of recursion variables,
$X,Y,\ldots$; $u$ ranges over \emph{identifiers}, which may be either
resource names or (value) variables.  We also assume a set of
\emph{principals} or \emph{owners} \Own containing at least two
elements, ranged over by $\ownfnt o, \ownfnt u, \ownfnt p$, who are
implicitly registered for these resources and who finance their
provision and use. The syntax of \picost is then given in
Figure~\ref{fig:syn}, and is essentially a very minor variation on
\DPI, \cite{dpibook}. The main syntactic category represents code
running under responsibility, with $\Cloc{P}{o}$ being the novel
construct. As explained in the Introduction this represents the code
$P$ running under the responsibility of the owner $\ownfnt o$;
intuitively $\ownfnt o$ is financially responsible for the computation
$P$. Thus in general a system is simply a collection of computation
threads each running under the responsibility of an explicit owner,
which may share private resources. The syntax for these threads is a
version of the well-known picalculus, \cite{pibook}.

The type $\tR$ of a  resource  describes the costs 
associated with that resource. There is a cost associated with \emph{using}
a resource, and a cost associated with \emph{providing} it; therefore
types take the form
\begin{math}
  \pair(k_u,k_p)
\end{math}
where $k_u,\;k_p$ are elements from some \emph{cost domain} \Kost.
Here we take \Kost simply to be $\nats$ ordered in the standard manner, 
but most of our results apply equally well to variations.

We employ the standard abbreviations associated with the picalculus,
and associated terminology. In particular we assume \emph{Barendregt's
  convention}\label{barendregt}, which implies that bound variables used in terms or
definitions are distinct, and different from any free variables in use
in the current context.  In Figure~\ref{fig:syn} meta-variable $v$
range over \emph{value expressions}, whose specification we omit; but
they include at least resource names $a \in \Chan$, variables $x$ from
$\Var$, and elements of $\Kost$. As usual we omit every
occurrence of a trailing $\Cstop$ and abbreviate $u?\pa{}.T,\,
u!\pc{}.T$ to $u?.T,\, u!.T$ respectively.  We are only interested in
\emph{closed code terms}, those which contain no free occurrences of
variables, which are ranged over by $P,\;Q,\ldots$; we use   
$\fn{P}$ to denote the set of names from  $\Chan$ which occur freely in $P$. 
 In the
sequel we assume all terms are closed.

\subsection{Cost environments:}
Since computations have financial implications, the execution of
processes is now relative to a \emph{cost environment} $\Gamma$. This
records the financial resources available to principals, and the cost
of providing and using resources; in order to be able to compare the
cost of computations we also assume a component which records the
expenditure as a computation proceeds.
% owners are registered for which resources, and both the costs required
% to use resources, and the effect of actually using them.  
Thus judgements of the reduction semantics take the form
\begin{align*}
  \Gamma \with M \;\smalleval\; \Delta\with N
\end{align*}
where $\Gamma,\;\Delta$ are cost environments.

There are many possibilities for \emph{cost environments}; see \cite{picost} 
for an example which directly associates funds with resources. In the present paper
we define them in such a way that the owners retain total control over their own funds. 

\begin{defi}[Cost environments]\label{def:costenv}
A  \emph{cost environment} $\Gamma$ consists
of a 4-tuple $\langle \Gamma^o,\Gamma^u, \Gamma^p,$ $\Gamma^{\record}\rangle$ where
\begin{enumerate}[$\bullet$]

\item 
$\Gamma^u: \Chan \parrow \Kost$

$\Gamma^u(a)$ records the cost of \emph{using} resource $a$; this is a
\emph{static} component, and will not vary during computations

\item 
$\Gamma^p: \Chan \parrow \Kost$

$\Gamma^p(a)$ records the cost of \emph{providing} resource $a$; again
this is a static component

\item 
$\Gamma^o: \Own \parrow \Kost$

$\Gamma^o({\ownfnt o})$ records the funds available to owner $\ownfnt
o$; this will vary as computations proceed, as owners will need to
fund their interactions with resources

\item $\Gamma^{\record} \in \Kost$ 

 $\Gamma^{\record}$ keeps an
account of the expenditure occurred during a computation; of course
this also will vary as a computation proceeds.
\end{enumerate}

\medspace We assume that both functions $\Gamma^u,\;\Gamma^p$ have the same
finite domain, but not necessarily that $\Gamma^u(a) \geq \Gamma^p(a)$
whenever these are defined. \boxHere
\end{defi}

We now define some operations on cost environments which will enable
us to reflect their impact on the semantics of our language. The most
important is a partial function, $\Gamma \ar{\cancom{u}{a}{p}}
\Delta$, which informally means that in $\Gamma$ owner $\ownfnt u$ has
sufficient funds to cover the cost of using resource $a$ and owner
$\ownfnt p$ has sufficient funds to provide it.  Then $\Delta$ records
the result of the expenditure of both $\ownfnt o$ and $\ownfnt p$ of
those funds. There is also considerable scope as to what happens to
these funds, and how their expenditure is recorded.  Here we take the
view that the provider $\ownfnt p$ gains the cost which the user
expends, to offset  $\ownfnt p$'s cost in providing the resource.
\begin{defi}[Resource charging]\label{def:rescharging}
  Let $\ar{\cancom{u}{a}{p}}$ be the partial function 
over cost environments defined as follows:
$\Gamma \ar{\cancom{u}{a}{p}} \Delta$ if
\begin{enumerate}[(i)]

%\item the owners $\ownfnt u$ and $\ownfnt p$ are different

\item $\Gamma^o({\ownfnt u}) \geq \Gamma^u(a)$ and 
$\Gamma^o({\ownfnt p}) \geq \Gamma^p(a)$

\item $\Delta$  is the cost environment obtained from $\Gamma$ by 
  \begin{enumerate}[(a)]
  \item decreasing $\Gamma^o({\ownfnt u})$ by the amount  $\Gamma^u(a)$

  \item increasing  $\Gamma^o({\ownfnt p})$ by the 
amount $\Gamma^u(a) - \Gamma^p(a)$, which may of course be negative

  \end{enumerate}

\item 
  Finally there is considerable flexibility in how this resource
  expenditure is recorded in $\Delta^{\record}$.  We call resource 
  charging  for $a$  \emph{standard} when this is set to $\Gamma^{\record} 
  + \Gamma^u(a) -  \Gamma^p(a)$; that is we add to the record the gain
  obtained in using resource $a$.  But in general we allow functions
  $\record_a(-,-)$, for each resource $a$, in which case we define
  $\Delta^{\record}$ to be $\Gamma^{\record} + \record_a(\Gamma^u(a),\Gamma^p(a))$. \boxHere
\end{enumerate}
\end{defi}
\noindent
In general we allow the owners $\ownfnt u$ and $\ownfnt p$ in this definition
to coincide. So, for example if 
$\Gamma \ar{\cancom{o}{a}{o}} \Delta$, then the effect of performing (a) above, 
followed by (b), is that 
$\Delta^o(\ownfnt o)$ is set to  $\Gamma^o(\ownfnt o) - \Gamma^p( a)$.
 
The use of two independent charges for each resource, $\Gamma^u$ and
$\Gamma^p$, may seem overly complex. A simpler model can be obtained by
having only one combined charge; effectively we could assume
$\Gamma^p(a)$ to be $0$ for every $a$, and so resource charging simply
transfers the appropriate amount of funds from the user to the
provider; this could be achieved by restricting attention to
\emph{simple types}, resource types $\tR$ of the form $\pair(k_u,0)$.
Indeed this simplification will be quite useful in order to achieve
some theoretical  properties of our proof methodology; see 
Definition~\ref{def:simpletypes} and Section~\ref{sec:fa}.
Nevertheless 
the use of the two independent charges $\Gamma^p(-)$ and $\Gamma^u(-)$ allows scope for more
interesting examples. In particular it provides considerable scope for variation 
in the manner in which resource expenditure is recorded in the component $\Gamma^{\record}$;
see Example~\ref{ex:publishing} for an instance.

\leaveout{
The restriction in (i) that the owners  $\ownfnt u$ and $\ownfnt p$ be 
different is slightly inconvenient, but as we will see it simplifies the
development of the theory somewhat. This simplification will be a consequence
of  the following result:
\begin{lem}\label{lem:1}
  Suppose $\Gamma \ar{\cancom{u}{a}{o_1}} \Delta$ and 
$\Gamma \ar{\cancom{o_2}{a}{p}} \Delta$. Then 
$\Gamma \ar{\cancom{u}{a}{p}} \Delta$. 
\end{lem}
\proof
Ah 
\qed
It is worth pointing out that this result is no longer true if condition
(i) is omitted from Definition~\ref{def:costenv}.
\begin{exa}\label{ex:1}
 A natural definition of $\ar{\cancom{o}{a}{o}}$ would be to allow
$\Gamma  \ar{\cancom{o}{a}{o}} \Delta$ whenever $\Delta$ is determined
by
\begin{align*}
  \Delta^o({\ownfnt x}) =
  \begin{cases}
    \Gamma^o({\ownfnt x}) - \Gamma^p(a)  & \text{if}\; {\ownfnt x} = {\ownfnt o} \\
     \Gamma^o({\ownfnt x})   &\text{otherwise}
  \end{cases}
\end{align*}
with the other components of $\Delta$ being determined 
as in Definition~\ref{def:costenv}.

With this definition Lemma~\ref{lem:1} is no longer true. Let $\Gamma$
be a cost environment with resource $a$ with associated costings $\Gamma^u(a) = 1$
and $\Gamma^p(a) = 2$, and suppose there are two owners, with funds given by
$\Gamma^o({\ownfnt o}) = 10,\; \Gamma^o({\ownfnt o_1}) = 20$.  Then there is a cost
environment $\Delta$ such that
\begin{enumerate}[$\bullet$]
\item $\Gamma \ar{\cancom{o}{a}{o_1}} \Delta$

\item $\Gamma \ar{\cancom{o_1}{a}{o}} \Delta$

\item but \textbf{not} $\Gamma \ar{\cancom{o}{a}{o}} \Delta$.
\end{enumerate}
$\Delta$ is determined by the fact that $\Delta^o({\ownfnt o}) = 9$ and 
$\Delta^o({\ownfnt o_1}) = 19$.  It follows that $\Gamma \ar{\cancom{o}{a}{o}} \Delta$
can not be true as it requires $\Delta^o({\ownfnt o})$ to be 
$\Gamma^o({\ownfnt o}) - \Gamma^p(a)$, namely $8$. \qed
\end{exa}

}
\noindent
We also need to extend cost environments with new resources.
\begin{defi}[Resource registration]\label{def:resourcereg}
  The cost environment $\Gamma, \Cassoc{a}{\tR}$, is only defined if
  $a$ is \emph{fresh} to $\Gamma$, that is, if $a$ is neither in
  $\dom{\Gamma^u}$ nor in $\dom{\Gamma^p}$. In this case it gives the new
  cost environment $\Delta$ obtained by adding the new resource, with
  the capabilities determined by $\tR$.  Formally the dynamic components
of $\Delta$, namely $\Delta^o$ and $\Delta^\record$, are inherited directly
from $\Gamma$, while the static components have the obvious definition; for 
example if $\tR$ is the type $\pair(k_u,k_p)$ then $\Delta^u$ is given by
\begin{align*}
  \Delta^u(x) &= \begin{cases}
                       k_u  & \text{if}\; x=a\\
                       \Gamma^u(x) &\text{otherwise}
                 \end{cases}
\end{align*}
We also assume that the  resource charging for $a$
  in $(\Gamma,\Cassoc{a}{\tR})$ is always standard. \boxHere
\end{defi}

\noindent
Note that every cost environment may be written in the form
\begin{align*}
  \Gammadyn, \Cassoc{a_1}{\tR_1},\ldots \Cassoc{a_n}{\tR_n}
\end{align*}
where $\Gammadyn$  is a \emph{basic} environment;
that is the static components $\Gammadyn^u$ and $\Gammadyn^p$
are both empty, and so it 
only contains  non-trivial dynamic components. 

\begin{figure}[t]
$
\begin{array}{l}
 \linfer[\Rred{comm}]
        {\Gamma \ar{\cancom{u}{a}{p}} \Delta}
        {\Gamma \with \Cloc{a!\pc{v}.Q}{u} \Cpar \Cloc{a?\pa{x}. P}{p}
         \smalleval
         \Delta \with \Cloc{Q \Cpar P\subst{x}{v}}{p} }
 \\
\slinfer[\Rred{split}]
       %{}
       {\Gamma \with \Cloc{M \Cpar N}{o}
         \smalleval
       \Gamma \with \Cloc{M}{o} \Cpar \Cloc{N}{o}}
\\
\slinfer[\Rred{export}]
       %{}
       {\Gamma \with \Cloc{\Cnew{r}{\tR}P}{o}
         \smalleval
       \Gamma \with \Cnew{r}{\tR}\Cloc{P}{o} }
\\

\slinfer[\Rred{unwind}]
     {\Gamma \with \Cloc{\Crec{x}{T}}{o} 
         \smalleval 
      \Gamma \with \Cloc{T\subst{x}{ \Crec{x}{T}}  }{o}  }
\\
 \slinfer[\Rred{match}]
         {\Gamma \with \Cloc{\Cmatch{a=a}{P}{Q}}{o}
          \smalleval
         \Gamma \with \Cloc{P}{o}}
\\
 \slinferSIDE[\Rred{mismatch}]
         {\Gamma \with \Cloc{\Cmatch{a=b}{P}{Q}}{o}
          \smalleval
         \Gamma \with \Cloc{Q}{o}}
         {a\not= b}
\\

 \linfer[\Rred{struct}]
         {M \structeq M',\; \Gamma  \with M \smalleval \Delta \with N,\;N \structeq N' }
         {\Gamma \with M' \smalleval \Delta \with N'}
\\
\linfer[\Rred{cntx}]
         {\Gamma \with M \smalleval \Delta \with M' }
         {\Gamma \with M\Cpar N \smalleval \Delta \with M'\Cpar N}
\\
  \linfer[\Rred{new}]
         {\Gamma, \Cassoc{b}{\tR} \with M
                \smalleval \Delta,\Cassoc{b}{\tR}  \with N }
         {\Gamma \with \Cnew{b}{\tR} M  \smalleval  
                         \Delta \with \Cnew{b}{\tR}N}
\end{array}
$

\caption{Reduction semantics \label{fig:reductions}}

\begin{displaymath}
\begin{array}{lrcl}
\Rstruct{extr}&  \Cnew{r}{\tR}(M \Cpar N)  &\structeq& M \; \Cpar \; \Cnew{r}{\tR}N,
\ \textrm{if}\; r \not\in \fn{M}
\\
\Rstruct{com}& M \Cpar N                       &\structeq& N \Cpar M
\\
\Rstruct{assoc}& (M \Cpar N) \Cpar O      &\structeq&  M \Cpar (N \Cpar O)
\\
\Rstruct{zero}& M \Cpar \Cnil                &\structeq& M 
\\
& \Cloc{\Cstop}{o} &\structeq& \Cnil\\
\Rstruct{flip}& \Cnew{r}{\tR}\Cnew{r'}{\tR'} M      &\structeq&
\Cnew{r'}{\tR'}\Cnew{r}{\tR} M
\end{array}
\end{displaymath}
\caption{Structural equivalence of $\picost$ \label{fig:structeq}}

\end{figure}

\subsection{Reduction semantics:}
The pair $(\Gamma \with M)$ is called a \emph{configuration} provided that
$\fn{M} \subseteq \dom{\Gamma^u} = \dom{\Gamma^p}$, that is every free
resource name in $M$ is known to the \emph{cost environment} $\Gamma$.
The reduction semantics for \picost is then defined as the least
relation over configurations which satisfies the rules in
Figure~\ref{fig:reductions}.  The majority of the rules come directly
from the reduction semantics of \DPI, \cite{dpibook}, and are
housekeeping in nature. The only rule of interest is \Rred{comm},
representing the communication along the channel $a$, or in \picost
the \emph{use} of the resource $a$ by owner $\ownfnt u$ which is
\emph{provided} by owner $\ownfnt p$. However this reduction is only possible
whenever the premise $\Gamma \ar{\cancom{u}{a}{p}} \Delta$ is
satisfied.
As we have seen, this means that in $\Gamma$ owner $\ownfnt u$
has sufficient funds to cover the cost of using resource $a$ and owner
$\ownfnt p$ has sufficient funds to provide it; and further $\Delta$ records
the result of the expenditure of both $\ownfnt u$ and $\ownfnt p$ of
those funds.

The remainder of the rules are borrowed directly from the standard
reduction semantics of \DPI; note that \Rred{struct} requires a
structural equivalence between terms; this again is the standard one
from \DPI, the definition of which is given  in
Figure~\ref{fig:structeq}. Also the final rule \Rred{new} uses the
registration operation on cost environments, given in
Definition~\ref{def:resourcereg}.
\begin{prop}
  If $(\Gamma_1 \with M_1)$ is a configuration and 
$(\Gamma_1 \with M_1) \smalleval (\Gamma_2 \with M_2)$ 
then $(\Gamma_2 \with M_2)$ is also a configuration.  
\end{prop}
\proof
Straightforward, by induction on the proof that 
$(\Gamma_1 \with M_1) \smalleval (\Gamma_2 \with M_2)$.
When handling the rule \Rred{struct} it uses the obvious fact that $M \structeq N$
implies that $M$ and $N$ have the same set of free names; this in turn means that
$M \structeq N$ implies $\Gamma \with M$ is a configuration if and only if
$\Gamma \with N$ is. 
\qed

The reductions of a configuration affect its cost environment, and
as a sanity check we can describe precisely the kinds of changes which
are possible:\vfill\eject

\begin{prop}\label{prop:red}
  Suppose $(\Gamma_1 \with M_1) \smalleval (\Gamma_2 \with M_2)$. Then 
  \begin{enumerate}[\em(i)]
  \item $\Gamma_1 = \Gamma_2$, and $(\Delta \with M_1) \smalleval (\Delta \with M_2)$ whenever $(\Delta \with M_1)$ is a
        configuration
  \item or $\Gamma_1 \ar{\cancom{u}{a}{p}} \Gamma_2$, for some resource $a$ and owners ${\ownfnt u,\;  \ownfnt p}$,
            and whenever $(\Delta \with M_1)$ is a
        configuration $\Delta \ar{\cancom{u}{a}{p}} \Delta'$ implies  $(\Delta \with M_1) \smalleval (\Delta' \with M_2)$
  \item or $\Gamma_1,\Cassoc{a}{\tR} \ar{\cancom{u}{a}{p}} \Gamma_2, \Cassoc{a}{\tR}$, for some (fresh) resource $a$, resource type
$\tR$  and owners ${\ownfnt u,\;  \ownfnt p}$, and whenever $(\Delta \with M_1)$ is a
        configuration
       $\Delta,\Cassoc{a}{\tR} \ar{\cancom{u}{a}{p}} \Delta', \Cassoc{a}{\tR}$ implies 
         $(\Delta \with M_1) \smalleval (\Delta' \with M_2)$
  \end{enumerate}
\end{prop}
\proof
Again this is a simple proof by rule induction on the premise
$(\Gamma_1 \with M_1) \smalleval (\Gamma_2 \with M_2)$.  Intuitively possibility (i) corresponds to a move
where no communication occurs, (ii) is when the move is a communication along a channel $a$ known to $\Gamma_1$,
and (iii) when the communication is along a private internal channel. 
\qed

\subsection{Examples:}\label{sec:examples1}

\newcommand{\Reader}{\pfn{Reader}}
\newcommand{\Store}{\pfn{Store}}
\newcommand{\Library}{\pfn{Library}}
\newcommand{\Sys}{\pfn{Sys}}
\newcommand{\UD}{\pfn{UD}}
\newcommand{\Lib}{\pfn{Lib}}
\newcommand{\GB}{\pfn{Book}}

\newcommand{\local}{\pfn{local}}
\newcommand{\central}{\pfn{central}}

\newcommand{\goLib}{\cfn{goLib}}
\newcommand{\goHome}{\cfn{goHome}}
\newcommand{\reqR}{\cfn{reqR}}
\newcommand{\reqS}{\cfn{reqS}}
\newcommand{\Cname}{\cfn{name}}
\newcommand{\Cbook}{\cfn{book}\xspace}
\newcommand{\Creq}{\cfn{req}}

\newcommand{\Cpublic}{\ownfnt{pub}}
\newcommand{\Clib}{\ownfnt{lib}}
\newcommand{\kate}{\ownfnt{kate}}
\newcommand{\dad}{\ownfnt{dad}}

\newcommand{\news}{\cfn{news}}
\newcommand{\publish}{\cfn{publish}}
\newcommand{\adv}{\cfn{adv}}

Formally \picost has only unary communication, but
in these examples we will informally allow the communication of tuples
along channels. In addition we will use the standard abbreviations
associated with the picalculus.  We also omit types for channels when
they are not relevant; in such cases we assume that they cost nothing
to provide, and that there is no charge for using them.  It will be
convenient to have an \emph{internal choice} operator, with $P \oplus
Q$ representing an internal choice between $P$ and $Q$.  This can be
taken to be short-hand notation for $\CnewNT{c}( c!\pc{} \Cpar
c?\pa{}.P \Cpar c?\pa{}. Q)$, where $c$ is a fresh channel.

\begin{figure}[t]
%  \centering
  
  \begin{align*}
\Sys       &\Leftarrow (\;\Cloc{\Reader}{\Cpublic} \;\Cpar\; 
                        \Cloc{\Library \Cpar \Store}{\Clib}\; ) \\
\text{where}\\
    \Reader &\Leftarrow \Crec{R} {\goLib?\pa{\Cname}. \CnewNT{r}\; \reqR!\pc{r,\Cname}. \\
            &\phantom{\Leftarrow \Crec{x} {\goLib?\pa{\Cname}.}}    r?\pa{b}. \goHome!\pc{b}. R}\\
   \Library &\Leftarrow \Crec{L}{ \reqR?\pa{y,z}. \phantom{\oplus} y!\pc{\Cbook(z)}. L\\
            &\phantom{\Leftarrow \Crec{x}{ \reqR?\pa{x,y}.}} \oplus 
                   \CnewNT{r} \;\reqS!\pc{r,z}. r?\pa{b}. y!\pc{b}. L}\\
  \Store    &\Leftarrow \Crec{S}{\reqS?\pa{y,z}. y!\pc{\Cbook(z)}. S}\\
  \end{align*}
  \caption{Running a library}
  \label{fig:lib}
\end{figure}

\begin{exa}[Running a library]\label{ex:lib}

Consider the system $\Sys$ from Figure~\ref{fig:lib}, which consists
of three recursive components, a library user $\Reader$, running under
the responsibility of the principal $\Cpublic$, standing for $\cfn{public}$, 
a library interface $\Library$ and an auxiliary book depository $\Store$, both running
under some other principal $\Clib$.

The programming of these components
involves the systematic generation of \emph{reply channels}. Thus for example 
the $\Reader$ gets the name of a book with which to go to the library, generates 
a new reply channel $r$ and submits  this together with the name of the book via 
$\reqR$; it awaits the book and then returns home.  The $\Store$ is also very simple;
it recursively awaits a request on $\reqS$, consisting of a reply channel and a $\Cname$
and returns the appropriate book on the channel. Finally the $\Library$
service requests at $\reqR$ consisting of a reply channel and a name. The book may be immediately
available, in which case it is returned, or it may be necessary to send a request to the $\Store$. 

Let us now consider the behaviour of these systems relative to two
cost environments $\Gamma_{\local}, \;\Gamma_{\central}$ representing
two different strategies for providing library services.  To focus on
the relative cost of providing these services let us assume that their
use is free, that is $\Gamma^u_{*}(a) = 0$ for every resource $a$, where $*$ ranges over
$\local,\,\central$, and that 
the amount of funds available is not an issue, that is 
$\Gamma_{*}^o(\Cpublic) = \Gamma_{*}^o(\Clib) = \infty$.  The
cost of providing the services, $\Gamma^{p}_{*}$ is given in the table
below, reflecting on the one hand the relative convenience to the $\Reader$ 
of the local services, and on the other the relative convenience to the
authorities in providing central services.

\begin{center}
\begin{tabular}{|l|| c | c|}
\hline
          &$\local$ &$\central$\\\hline
$\goLib$    & 1       &5\\\hline
$\goHome$   & 1       &5\\\hline
$\reqR$     & 3       &1\\\hline
$\reqS$      & 5       &1\\\hline
\end{tabular}
\end{center}
Finally let us take the counters
$\Gamma^{\record}_{*}$ to be initially set to $0$. 
Note that $\Gamma_{\local}$ can be written as
\begin{align*}
  \Gammadyn,\; \Cassoc{\goLib}{\tR_l^g},\;\Cassoc{\goHome}{\tR_l^h},\;\Cassoc{\reqR}{\tR_l^r},\;\Cassoc{\reqS}{\tR_l^s}
\end{align*}
where $\tR_l^g,\tR_l^h, \tR_l^r, \tR_l^s$ are the types $\pair(0,1),\pair(0,1), \pair(0,3), \pair(0,5)$ respectively,
and $\Gammadyn$ is a basic environment;  $\Gamma_{\central}$ has a similar representation, with a slightly different
sequence of types.

To exercise the system we use
\begin{align*}
  \GB \Leftarrow \Cloc{\goLib!\pc{str}. \goHome?\pa{x}. \Cstop}{\Cpublic}
\end{align*}
to prod the $\Reader$ into action, 
where $str$ is the name of some book. 
Consider the configuration 
\begin{align*}
  \calC_1 = \Gamma_{\local} \with (\GB \Cpar \Sys),
\end{align*}
and let us ignore the computation steps involved in generating reply
channels, and general housekeeping such as the unwinding of recursive
definitions, which in any event cost nothing.  Because of the internal
non-determinism in the library service there are essentially two
computations from $\calC_1$. If the $\Store$ is not used then after
three computation steps which require funds it is in the state $ \Delta_{\local} \with
\Sys $, where $\Delta_{\local}^{\record} = 5$. This represents the overall cost of
this transaction, $2$ of which is paid by $\Cpublic$ and $3$ by $\Clib$. 

On the other hand if the $\Store$   is used, then there are 
four computation steps which require funding, after which the state 
$\Theta_{\local}
\with \Sys$ is reached, where $\Theta_{\local}^{\record} = 10$.  However using
the central cost environment $\Gamma_{\central}$ the two possibilities
are $\Delta_{\central}^{\record} = 11$ and
$\Theta_{\central}^{\record} = 12$ respectively.  In each eventuality
the local implementation is more efficient, in the sense that the costs are systematically
lower. 
\boxHere
\end{exa}

The charging regime for resources is such that their use effectively
means a transfer of funds to the \emph{provider} from the \emph{user},
provided the cost of providing the resource is less than the charge
for its use. This enables us to implement a systematic way of
transferring funds between owners.
\begin{exa}[Fund transfer]\label{ex:fund.transfer} 
Consider the systems defined as follows:
  \begin{align*}
    \Sys &\Leftarrow  \Cloc{D}{\dad} \;\Cpar \Cloc{K}{\kate}\\
\text{where}\\
D        &\Leftarrow \Creq?\pa{x}. \Cnew{s}{\tR_s} x!\pc{s}. s!.S\\
K        &\Leftarrow \CnewNT{r} \Creq!\pc{r}.r?\pa{y}. y?.H
  \end{align*}
The size of the transfer from $\dad$ to $\kate$ depends on the type $\tR_s$ at which the
new channel $s$ is declared. Suppose this type is $\pair(0,k)$, and let $\Gamma$ be a cost environment
in which $\Gamma^o(\dad)$ is at least $k$.  Then there is a computation
\begin{align*}
  (\Gamma \with \Sys) \;\smalleval^*\; (\Delta \with \Cloc{S}{\dad} \;\Cpar \Cloc{H}{\kate})
\end{align*}
in which $\Delta^o(\dad) = \Gamma^o(\dad) - k$ and $\Delta^o(\kate) = \Gamma^o(\kate)+k$.
\boxHere
\end{exa}
\begin{figure}[t]
%  \centering
  
  \begin{align*}
\Sys       &\Leftarrow \Cloc{P}{p} \;\Cpar\; 
                        \Cloc{N}{n} \;\Cpar\;  \Cloc{A}{a} \;\Cpar\; \Cloc{R}{r}\\
\text{where} \\
P          &\Leftarrow \Crec{P}{\CnewNT{r_1} \news!\pc{r_1}. \CnewNT{r_2} \adv!\pc{r_2}. \\
           &\phantom{\Leftarrow  \CnewNT{r_1} }  r_1?\pa{n}. r_2?\pa{d}. \publish?\pa{z}. z!\pc{n,d}. P  }
\\
N         &\Leftarrow \Crec{N}{\news?\pa{r} \CnewNT{n} r!\pc{n}.N}
\\
A         &\Leftarrow \Crec{A} {\adv?\pa{r}. \CnewNT{d} r!\pc{d}.A}
\\
R         &\Leftarrow \Crec{R}{\CnewNT{r} \publish!\pc{r}. r?\pa{n,d}.R}
\end{align*}
  \caption{Publishing}
  \label{fig:publish}
\end{figure}

\begin{exa}[Publishing]\label{ex:publishing}

  Consider the system $\Sys$ in Figure~\ref{fig:publish}, which has four components:
  \begin{enumerate}[(a)]
  \item publisher:  \emph{uses} a news service via the resource $\news$,
        \emph{uses} an advertising agency via the resource $\adv$ and \emph{provides} 
         the resource  $\publish$ 

 \item news service: \emph{provides} a service via $\news$

 \item ad agency: \emph{provides} a service via $\adv$

 \item reader: \emph{uses} the resource $\publish$
  \end{enumerate}
The viability of publishing depends of course on the cost associated with these resources.
As an example consider an environment $\Gamma_{327}$, of the form 
$\Gammadyn, 
\Cassoc{\news}{\tR_n},
\Cassoc{\adv}{\tR_a},
\Cassoc{\publish}{\tR_p}$, where these types are
$\pair(3,1),\pair(2,0),\pair(7,1)$ respectively, and let us assume 
$\Gamma_{327}^{\record}$ is initialised to $0$. 
Furthermore, since we are concentrating on the publisher, let us assume that the 
resource charging is defined so that only the effect on the owner $\ownfnt p$ is 
recorded. Refering to Definition~\ref{def:rescharging} this means that resource charging
is standard for $\publish$ but 
we need to set $\record_{a}(k_u,k_p)$ to be $-k_u$, if
$a$ is either $\news$ or $\adv$. 

Now consider a computation from the configuration $\Gamma_{317} \with \Sys$. 
Provided the owners have sufficient funds, specifically $\Gamma^o(\ownfnt p), 
\Gamma^o(\ownfnt n)$ and $\Gamma^o(\ownfnt r)$ must be at least $5,1,7$ respectively,
then we have a computation
\begin{align*}
  (\Gamma_{317} \with \Sys)  \smalleval^* (\Delta_1 \with \Sys)
\end{align*}
where $\Delta_1^\record = 1$; the record part of the initial
environment was set to $0$, during the computation it was set to $-3$
after the publisher uses the $\news$ resource, then to $-5$ after
using $\adv$; finally, when the reader uses the $\publish$ resource,
this is increased by $(7-1)$ to give $1$.  Because we have defined
expenditure recording to reflect the point of view of the publisher, this
represents the fact that the publisher has made a profit of $1$ as a
result of this sequence of transactions.  Note also that at this point
$\Delta_1^o({\ownfnt p})$ is $\Gamma_{327}^o({\ownfnt p}) + 1$.

We can also see what happens when the costs of using resources is changed. 
Let $\Gamma_{216}$ be the environment in which  the cost of all three resources are decreased 
by $1$. Then we have the computation
\begin{align*}
  (\Gamma_{216} \with \Sys)  \smalleval^* (\Delta_2 \with \Sys)
\end{align*}
where now $\Delta_2^\record = 2$; this represents an increase in profits for the publisher. 
\boxHere
\end{exa}

\begin{exa}[Kickbacks]\label{ex:kickback}

Suppose in Figure~\ref{fig:publish} we change the situation so that the publisher obtains a kickback
from the ad agency when an ad is downloaded. The modified code
is given by
\begin{align*}
  P_K         &\Leftarrow \Crec{P}{\CnewNT{r_1} \news!\pc{r_1}. \CnewNT{r_2} \Cnew{k}{\tK} \adv!\pc{k,r_2}. \\
           &\phantom{\Leftarrow  \CnewNT{r_1} }  r_1?\pa{n}. r_2?\pa{d}. \publish?\pa{z}. k?. z!\pc{n,d}. P  }
\\
A_K         &\Leftarrow \Crec{A} {\adv?\pa{k,r}. \CnewNT{d} r!\pc{d}.(A \Cpar k!)}
\end{align*}
and let $\Sys_K$ denote the revised system. The size of the kickback depends on the parameters in the 
type $\tK$. In $\Sys$ the ad agency receives the benefit $2$ for supplying the ad; if we set $\tK$ to be
$\pair(1,0)$ then in $\Sys_K$ this benefit is split equally with the publisher. 
Under the same assumptions as in Example~\ref{ex:publishing} we have the computations
\begin{align*}
  (\Gamma_{327} \with \Sys_K)  \smalleval^* (\Phi_1 \with \Sys_K)\quad\quad\text{and}\qquad
    (\Gamma_{216} \with \Sys_K)  \smalleval^* (\Phi_2 \with \Sys_K)
\end{align*}
where now $\Phi_1^\record, \Phi_2^\record$ are  $2,3$ respectively, indicating more profit in each case
for the publisher. 
\boxHere
\end{exa}
\section{Compositional reasoning}\label{sec:lts}

The aim of this section is to develop a proof methodology 
for \picost. The idea is to define a  \emph{behavioural preorder} 
\begin{align}\label{eq:eff}
(\Gamma \with M) \sqsubseteq (\Delta \with N),  
\end{align}
meaning that in some sense $ (\Gamma \with M) $ and $ (\Delta \with
N)$ offer the same behaviour, but the latter is at least as efficient
as the former, and possibly more.  We follow the standard approach
of defining the preorder (\ref{eq:eff}) as the largest relation
between \picost configurations satisfying a transfer property,
associated with the ability of processes to interact with their
peers. We thereby automatically get a co-inductive proof methodology
for establishing relationships between configurations.

In fact, referring to (\ref{eq:eff}), it is better to move away from
terminology such as \emph{efficiency} as the interpretation 
depends very much on the nature of the units being
recorded. In Example~\ref{ex:lib} these are \emph{costs} and in such a
scenario it is reasonable to interpret (\ref{eq:eff}) as saying
$(\Delta \with N)$ is an improvement on $(\Gamma \with M)$ as it
potentially involves less cost. On the other hand in
Example~\ref{ex:publishing} the units are \emph{profit} (for the publisher), and here
$(\Gamma \with M)$ would be considered to be an improvement on
$(\Delta \with N)$, as there is potential for more profit (for the publisher).

We therefore move to the more neutral terminology of \emph{weights}.
However we can not simply base the formulation of (\ref{eq:eff}) on
the relative weight associated with each individual action, as the
following example shows.  \newcommand{\Up}{\cfn{up}}
\newcommand{\down}{\cfn{down}}

\begin{exa}[Amortising costs]\label{ex:amort}
Consider the simple system 
\begin{align*}
  \UD \Leftarrow  \Cloc{\Crec{x}{\Up!.\down!.x}}{o}
\end{align*}
and let $\Gamma_{25}$ be an environment in which the unique owner
$\ownfnt o$ has unlimited funds, the use of $\Up$ costs $2$ and the
use of $\down$ costs $5$. If we compare $(\Gamma_{25} \with \UD)$ with
$(\Gamma_{42} \with \UD)$, where $\Gamma_{42}$ is defined analogously,
then intuitively the latter is more efficient than the former, despite
the fact that in the latter the action $\Up$ is more expensive; this
is compensated for by the relative costs of the other action $\down$.
\boxHere
\end{exa}

The remainder of this section is divided into three subsections. In the
first we present a theory of amortised weighted bisimulations, based on
so-called \emph{weighted labelled transition systems}, wLTSs.  This gives rise
to a parametrised behavioural preorder, which we call the
\emph{amortised weighted bisimulation preorder}.  The aim is to apply
this theory to $\picost$; with this in mind, in the second subsection
we present a (detailed) labelled transition semantics for \picost, and
show that it is in agreement with the reduction semantics given in
Figure~\ref{fig:reductions}. In the third section we show how this 
automatically generates a wLTS, which in turn gives us an
amortised weighted bisimulation preorder between \picost configurations. 
We demonstrate the usefulness of the resulting proof methodology
by re-examining the examples from Section~\ref{sec:examples1}.

\subsection{Amortised weighted bisimulations:}

Here we generalise  the concepts of \cite{astrid}; our aim is to apply them 
to \picost but our formulation is at a more abstract level.
\begin{defi}[Weighted labelled transition systems]
An \emph{weighted labelled transition system} or wLTS is a 4-tuple 
$\langle S, \Act_\tau, W,\ar{} \rangle$ where $S$ is a set of states,
$W$ set of weights, and 
$\ar{} \;\subseteq S \times \Act_{\tau} \times W \times S$.
Here $\Act_\tau$ denotes a set of action names  $\Act$ to which is added an
extra distinct name $\tau$ which will represent internal action. 
We normally write $s \ar{\mu}_w s'$ to mean 
$(s,\mu,w,s') \in \ar{}$. As a default we take the set of weights
to be $\bbb Z$, the set of integers, both negative and positive. \boxHere
\end{defi}
A wLTS is called \emph{standard} whenever there is a cost function
$\costN: \Act \rightarrow W$ with the property that $s \ar{a}_w s'$ if
and only if $w= \cost{a}$ for every $a \in \Act$.  So in a standard
wLTS there is a unique weight associated with external actions,
although internal actions may have multiple possible associated
weights, reflecting the different ways in which these actions may be
generated from external moves. The wLTS which we will (eventually)
generate for \picost will be standard, but the development below will
not require that we are working with standard wLTSs.

Relative to a given wLTS \emph{weak moves} are generated in the
standard manner, although the associated weights need to be
accumulated: $s \dar{\mu}_w s'$ is the least relation satisfying:
\begin{enumerate}[$\bullet$]
\item $s \ar{\mu}_w s'$ implies  $s \dar{\mu}_w s'$

\item  $s \dar{\mu}_w s'',\; s'' \ar{\tau}_v s'$ implies $s \dar{\mu}_{(w+v)} s'$

\item  $s \ar{\tau}_w s'',\; s'' \dar{\mu}_v s'$ implies $s \dar{\mu}_{(w+v)} s'$
\end{enumerate}
We also use a variation on the standard notation $s \dar{\hat{\mu}}_w t$ from  \cite{ccs};
when $\mu$ is any action other than $\tau$ this denotes  $s \dar{\mu}_w t$, but when it is 
$\tau$ it means either that $s \dar{\tau}_w t$ or that $s$ is $t$ and $w = 0$.

\begin{defi}[Amortised weighted bisimulations]\label{def:amort}
  A family of relations $\setof{\calR^n}{n \in \nats} $ over the states in
  a wLTS is called an \emph{amortised weighted bisimulation} whenever $ s \calR^n
  t$:
\begin{enumerate}[(i)]
\item $s \ar{\mu}_v s'$ implies $t \dar{\hat{\mu}}_w t'$ for some $t', w$ such that
    $s' \calR^{(n+v-w)} t'$

\item conversely, 
 $t \ar{\mu}_w t'$ implies $s \dar{\hat{\mu}}_v s'$ for some $s', v$ such that
    $s' \calR^{(n+v-w)} t'$  \boxHere
\end{enumerate}
\end{defi}
\noindent
Here the parametrisation with respect to \nats puts an extra
requirement on the standard \emph{transfer properties} associated with
bisimulations. In (i) and (ii) above the index $(n+v-w)$ must be in
\nats, that is must be non-negative. So for example if the
amortisation $n$ is 0 then $v$, the weight of the left hand action,
must be greater than or equal to $w$, the weight of the right hand
action. For this reason a standard bisimulation, which ignores the
weights, may not be an amortised weighted bisimulation. But the more
general effect of the parameter $n$ in the definition is to allow a
relaxation in the comparison between the actual weights of the
actions in the processes being compared; this point is explained in
detail in Example~\ref{def:amort.more}.

We can mimic the standard development of bisimulations and
write $s \amort{m} s'$ to say that there is some amortised
bisimulation $\setof{ \calR^n}{ n \in \nats}$ such that $ s \calR^m s'$.
Weighted bisimulations are (point-wise) closed under unions, and therefore we can
mimic the standard development of \emph{bisimulation equivalence}, \cite{ccs},
to obtain the following:
\begin{prop}\label{prop:amort.prop}\hfill
  \begin{enumerate}[\em(a)]
  \item The family of relations $\setof{ \amort{n}}{ n \in \nats}$ is an 
amortised  weighted bisimulation.
  \item This family is the largest (point-wise) amortised weighed bisimulation.
  \item If $s \amort{m} t$ and $s \dar{\mu}_v s'$ then $t \dar{\hat{\mu}}_w t'$ for some
        $t', v$ such that $s' \amort{(m+v-w)} t'$.
  \end{enumerate}
\end{prop}
\proof
Straightforward, using standard techniques.
\qed
When we are uninterested in the exact amortisation used we write simply
$s \amort{} t$, meaning that there is some $k \geq 0$ such that $s \amort{k} t$,
and we refer to this preorder as the \emph{amortised weighted bisimulation preorder}.
\begin{prop}\hfill
  \begin{enumerate}[\em(a)]
  \item  The relations $\amort{n}$ are reflexive 
  \item $s_1 \amort{m} s_2,\; s_2 \amort{n} s_3$ implies $s_1 \amort{(m+n)} s_3$
  \item $ \amort{m} \;\;\subseteq\;\; \amort{n}$ whenever  $m \leq n$. 
  \end{enumerate}
\end{prop}
\proof
In each case it is sufficient to exhibit a suitable amortised weighted bisimulation, that
is a suitable family of relations over states. For example to prove (b) we let
$\calR^k$, for $k \geq 0$, be the set of pairs $\pair(s_1,s_2)$ such that
$s_1 \amort{n} s_3$ and $s_3 \amort{m} s_2$
for some state $s_3$ and some numbers $n,m$ such that $k = n+m$. 

To show $\setof{ \calR^k}{ k \in \nats}$ is an amortised weighted
bisimulation let us suppose $s_1 \calR^k s_2$ and $s_1 \ar{\mu}_v
s'_1$; we have to prove
\begin{align}\label{req:1}
  s_2 \dar{\hat{\mu}}_w s_2' \;\text{for some $s'_2$ satisfying}\; 
s'_1 \calR^{(k+v-w)} s'_2
\end{align}
(The proof of the symmetric requirement is similar.)
\begin{enumerate}[(i)]
\item From $s_1 \amort{n} s_3$ we know $s_3 \dar{\hat{\mu}}_u s'_3$ such that 
$s'_1 \amort{(n+v-u)} s'_3$

\item From $s_3 \amort{m} s_2$, and the final part of the previous Proposition,
we know $s_2 \dar{\hat{\mu}}_w s'_2$ such that 
$s'_3 \amort{(m+u-w)} s'_2$.
\end{enumerate}
But since $(n+v-u) + (m+u-w) = (k+v-w)$ we have
$s'_1 \calR^{(k+v-w)} s'_2$ and the requirement (\ref{req:1}) follows. 

The proof of part (c) is similar using the family of relations
$\setof{\calR^n}{n \in \nats}$, where $s \calR^n t$ whenever $s \amort{m}
t$ for some $m \leq n$, while the proof of part (a) uses the family
where each $\calR^n$ is the identity relation.

 \qed
\begin{exa}[Amortising costs continued]\label{def:amort.more}
  Here we continue with Example~\ref{ex:amort}. Shortly we will see a
  systematic way of associating weights with actions in \picost. But
  informally we can simply say
\begin{align*}
  \calC_{25} \;\; \ar{\Up!}_2 \;\calD_{25}\; \ar{\down!}_5 \;\;  \calC_{25} 
\end{align*}
where $\calC_{25}$ and $\calD_{25}$ are abbreviations for the configurations 
$ (\Gamma_{25} \with \UD)$, respectively,
$(\Gamma_{25} \with \Cloc{\down!.\Crec{x}{\Up!.\down!.x}}{o})$, 
and analogously for $(\Gamma_{42} \with \UD)$. Then relative to this induced 
wLTS we can show that the following is a weighted bisimulation:
\begin{align*}
  \calR^n &= \sset{\pair(\calD_{25},\calD_{42})} \cup 
             \setof{\pair(\calC_{25},\calC_{42})}{n \geq 2}
\end{align*}
It follows that
\begin{align*}
   (\Gamma_{25} \with \UD)  \amort{2} (\Gamma_{42} \with \UD)  
\end{align*}

However 
\begin{math}
  (\Gamma_{42} \with \UD)  \not{\amort{k}}  (\Gamma_{25} \with \UD) 
\end{math}
for any $k$. To see this suppose $\setof{\calR^n}{n \geq 0}$ is a weighted
bisimulation; we prove by induction on $k$ that
\begin{align}\label{eq:fixup}
\pair(\calD_{42}, \calD_{25}) &\not\in \calR^{(k+2)}\\
\pair(\calC_{42}, \calC_{25}) & \not\in \calR^k \notag
\end{align}
First notice that the pair $\pair(\calD_{42}, \calD_{25})$ can not be in
$\calR^2$; this is because the move 
\begin{math}
  \calD_{42} \ar{\down!}_2 \calC_{42}
\end{math}
can not be matched by a move
\begin{math}
    \calD_{42} \dar{\down!}_w \calC_{42}
\end{math}
such that 
\begin{math}
  \calC_{42} \calR^{(2 +2 - w)} \calC_{25}.
\end{math}
The only only possible candidate is the move
\begin{math}
    \calD_{42} \dar{\down!}_5 \calC_{42}
\end{math}
and $\calR^{-1}$ does not exist. 

From this fact it follows immediately that the pair
$\pair(\calC_{42},\calC_{25})$ can not be in $\calR^0$; 
for matching the move 
\begin{math}
  \calC_{42} \ar{\Up!}_4 \calD_{42}
\end{math}
would require the impossible, that  $\pair(\calD_{42}, \calD_{25})$ be
$\calR^2$. In other words we have shown (\ref{eq:fixup}) in the case 
when $k=0$.

Suppose it is true for $k$; the proof that it follows for $(k+1)$ is
also straightforward. This is because
\begin{enumerate}[$\bullet$]
\item for $\pair(\calD_{42},\calD_{25})$ to be in $\calR^{(k+3)}$ we would
require that $\pair(\calC_{42},\calC_{25})$ be in $\calR^{(k+3+2-5)}$
which contradicts the induction hypothesis

\item for $\pair(\calC_{42},\calC_{25})$ to be in $\calR^{(k+1)}$ we would
require  $\pair(\calD_{42},\calD_{25})$ to be in $\calR^{(k+3)}$, which we have
just shown not to be possible.
\end{enumerate}

It  is important that the set of natural numbers $\nats$ is used in
Definition~\ref{def:amort}, or at least that the family of relations
be parametrised relative to a well-founded order. If instead we
allowed families of relations $\setof{R^z}{z \in {\bbb Z}}$, where ${\bbb Z}$ is the
set of all integers, positive and negative, then
\begin{math}
  (\Gamma_{42} \with \UD)  \amort{0}  (\Gamma_{25} \with \UD) 
\end{math}
would follow. Simply letting 
\begin{math}
  \calR^z = \sset{\pair(\calC_{42},\calC_{25}),\, \pair(\calD_{42},\calD_{25})}
\end{math}
for every $z \in {\bbb Z}$, we would obtain an extended family of relations trivially
satisfying the requirements in Definition~\ref{def:amort}.
Indeed in general, using ${\bbb Z}$ in place of \nats, there would be no difference between 
amortised weighted bisimulations and standard bisimulations (where all weights are ignored). 
 \boxHere
\end{exa}

\subsection{An operational semantics for  \picost}

As a first step in applying the theory of \emph{amortised weighted bisimulations} to 
\picost we give an operational semantics for the language in terms of a 
(standard) LTS. 

In Figure~\ref{fig:lts} and Figure~\ref{fig:ltsB} we
give a set of rules for deriving judgements of the form
$$
(\Gamma \with M) \newar{\lambda} (\Delta \with N),
$$ 
where $\lambda$ can take one of the forms 
\begin{enumerate}[(i)]
\item  internal action, $\tau$
\item input,  $\cancom{u}{(\Cassoc{\tilde{r}}{\tilde{\tR}}) a?v}{p}$:  
input by resource $a$ of  a known or fresh name, or value, where $\ownfnt p$ is the provider
of the resource and $\ownfnt u$ the user

\item output: $\cancom{u}{(\Cassoc{\tilde r}{\tilde \tR})a!{v}}{p}$: 
delivery of a known or fresh name, to resource $a$, where again $\ownfnt p$ 
is the provider of the resource and $\ownfnt u$ the user.

%\item external consumption, $\tau_a$: use by some external entity of resource $a$.
\end{enumerate}
We restrict attention to well-formed $\lambda$, that is, in the input and output actions
each $r_i$ must occur somewhere in $v$, and applications of the rules must preserve
well-formedness. 
However note that because \picost only uses unary communication the vectors 
$\tilde{(r)}, \tilde{(b)}$ will have length either 0 or 1. 
%We also use $\alpha$ to range over the  external actions, input and output,
%and in the rules we employ the standard \emph{complementary} notation 
%for them, $\overline{\alpha}$ denoting the complement of $\alpha$.

The rules are inherited directly from the corresponding ones for Dpi,
\cite{dpibook}, and for the sake of clarity obvious symmetric rules, such as for 
\Rlts{comm} and \Rlts{cntx}, are omitted; \emph{Barendregt's convention}
is also liberally applied, for example in omitting side-conditions to
\Rlts{cntx}. The only point of interest is the use of the
preconditions $\Gamma \ar{\cancom{o_1}{a}{o_2}} \Delta$ in \Rlts{in}
and \Rlts{out}; communication is only deemed to be possible if it can
be paid for in some manner. 
%For convenience we have omitted the obvious dual to the rule \Rlts{comm}. 
Note that $\ownfnt u$ in \Rlts{in}, and
$\ownfnt p$ in \Rlts{out} are free meta-variables. So for example the
simple process $\Cloc{a!\pc{v}.P}{o}$ can perform the actions
\begin{math}
   \Cloc{a!\pc{v}.P}{o}
          \newar{\cancom{o}{a!v}{o'}}
          \Delta \with \Cloc{P}{o}
\end{math}
for every owner $\ownfnt {o'} \in \Own$ such that $\Gamma \ar{\cancom{o}{a}{o'}} \Delta$.
Also in the communication rule \Rlts{comm} any new resources used in the communication,
$\tilde{r}:\tilde{\tR}$ remain private but in general the resulting cost environment $\Delta$ 
will be different from $\Gamma$; the internal communication involves the use of a resource,
and the change from $\Gamma$ to $\Delta$ will reflect the associated costs. 

\begin{figure}[t]
$
\begin{array}{ll}

 \linferSIDE[\Rlts{in}]
       {\Gamma \ar{\cancom{u}{a}{o}} \Delta}
       {\Gamma \with  \Cloc{a?\pa{x}.P}{o}  \newar{\cancom{u}{a?v}{o}}  
         \Delta \with \Cloc{P\subst{x}{v}}{o}}
       {v \in \dom{\Gamma^u} \;\text{or}\; v \; \text{not a channel}}
\\\\
 \linferSIDE[]
       {\Gamma \ar{\cancom{u}{a}{o}} \Delta }
       {\Gamma \with  \Cloc{a?\pa{x}.P}{o}  \newar{\cancom{u}{(b:\tR)a?b}{o}}  
         \Delta,b:\tR \with \Cloc{P\subst{x}{b}}{o}}
       {b \not\in \dom{\Gamma^u}}

  \\\\

% \linferSIDE[\Rlts{asy}]
%        {}%
%        {\Gamma \with \Cloc{P}{o} \ar{a?b}  
%          \Gamma \with \Cloc{P \Cpar \Cdel{a}{b}}{o}}
%        {b \in \dom{\Gamma^c}}
% \\\\
% \linferSIDE[]
%        {}%
%        {\Gamma \with \Cloc{P}{o} \ar{(b:\tR)a?b}  
%          \Gamma,b:\tR \with \Cloc{P \Cpar \Cdel{a}{b}}{o}}
%        {b \not\in \dom{\Gamma^c}}

%   \\\\
% \linfer[\Rlts{del}]
%          {\Gamma \ar{\cancom{o}{a}{p}} \Delta }
%         {\Gamma \with \Cloc{\Cdel{a}{v}}{o}
%          \arOwn{\Cdel{a}{v}}{O}
%          \Delta \with \Cloc{\Cstop}{o}  }
% \\\\
\linfer[\Rlts{out}]
 {\Gamma \ar{\cancom{o}{a}{p}} \Delta }
         {\Gamma \with  \Cloc{a!\pc{v}.P}{o}
          \newar{\cancom{o}{a!v}{p}}
          \Delta \with \Cloc{P}{o}
         }
\\\\
\linfer[\Rlts{comm}]
        {\Gamma \with M \newar{\cancom{u}{({\tilde{r}}:{\tilde{\tR}}) a?v}{p}} 
         \Delta,\Cassoc{\tilde{r}}{\tilde{\tR}} \with M',\;\; 
         \Gamma \with N
         \newar{\cancom{u}{ ({\tilde r}:{\tilde \tR})a!{v} }{p}} 
         \Delta,\Cassoc{\tilde{r}}{\tilde{\tR}} \with N'}
        { \Gamma \with M \Cpar N \newar{\tau}  
          \Delta \with \CnewNT{\, \Cassoc{\tilde{r}}{\tilde{\tR} } }(M' \Cpar N')     } 

\end{array}
$
\caption{An action semantics  for \picost: main rules\label{fig:lts}}

\hrulefill

\end{figure}
\begin{figure}
$
\begin{array}{ll}
\linferSIDE[\Rlts{open}]
       {\Gamma, \Cassoc{b}{\tR} \with M \newar{\cancom{u}{a!b}{p}} \Gamma' \with M'}
       { \Gamma \with \Cnew{b}{\tR} M \newar{\cancom{u}{(b:\tR)a!b}{p}} 
        \Gamma' \with  M'}
       {a \not =b}
&
\slinfer[\Rlts{export}]
       %{}
       {\Gamma \with \Cloc{\Cnew{r}{\tR}P}{o}
         \newar{\tau}
       \Gamma \with \Cnew{r}{\tR}\Cloc{P}{o} }
 \\\\
\slinfer[\Rlts{split}]
       %{}
       {\Gamma \with \Cloc{M \Cpar N}{o}
         \newar{\tau}
       \Gamma \with \Cloc{M}{o} \Cpar \Cloc{N}{o}}
&
\slinfer[\Rlts{unwind}]
     {\Gamma \with \Cloc{\Crec{x}{T}}{o} 
         \newar{\tau}
      \Gamma \with \Cloc{T\subst{x}{ \Crec{x}{T}}  }{o}  }
\\\\

 \linfer[\Rlts{match}]
         {}
         {\Gamma \with \Cloc{\Cmatch{a=a}{P}{Q}}{o}
          \newar{\tau}
         \Gamma \with \Cloc{P}{o}}
&
 \linferSIDE[\Rlts{mismatch}]
         {}
         {\Gamma \with\Cloc{ \Cmatch{a=b}{P}{Q}}{o}
          \newar{\tau}
         \Gamma \with \Cloc{Q}{o}}
         {a\not= b}
\\\\
\\

\linferSIDE[\Rlts{cntx}]{\Gamma \with M \newar{\lambda} 
                         \Gamma' \with M'}
             {\Gamma \with M \Cpar N \newar{\lambda}
               \Gamma' \with M' \Cpar N}
             {}

&
\linferSIDE[\Rlts{cntx}]{\Gamma,\Cassoc{b}{\tR} \with M \newar{\lambda} 
                         \Gamma',\Cassoc{b}{\tR} \with M'}
             {\Gamma \with \Cnew{b}{\tR}M \newar{\lambda}
                \Gamma' \with \Cnew{b}{\tR}M'}
             { b \not\in \names{\lambda}}
\\\\
 \end{array}
$
\caption{An action semantics  for \picost: more rules\label{fig:ltsB}}

\hrulefill
\end{figure}

We can perform a number of sanity checks on these rules. For example
one can show that if $(\Gamma_1 \with P_1) \newar{(b:\tR)\alpha} (\Gamma_2 \with
P_2)$ then $\Gamma_2 = \Delta,\Cassoc{b}{\tR}$ for some $\Delta$ such 
that $\Gamma_1 \ar{\cancom{u}{a}{p}} \Delta$, for some $\ownfnt u, \ownfnt p$, 
where $a$ is the channel used in $\alpha$; a more detailed analysis of the possible
judgements is given in the two lemmas below.
The actions also preserve configurations:
\begin{prop}
    If $(\Gamma_1 \with M_1)$ is a configuration and 
$(\Gamma_1 \with M_1) \newar{\lambda} (\Gamma_2 \with M_2)$ 
then $(\Gamma_2 \with M_2)$ is also a configuration. 
\end{prop}
\begin{proof}
  A straightforward induction on the inference of the judgements.
\end{proof}
We also have a consistency check with respect to the reduction semantics
of Section~\ref{sec:lang}, stated in the theorem below; the proof requires
two technical lemmas. 
\begin{lem}[Deriv-output]\label{lemma:derout1}
  Suppose $\Gamma \with M \newar{\cancom{u}{(\tilde{r}:\tilde{\tR})a!v}{p}} \Delta \with N$.
Then
\begin{enumerate}[\em(i)]
\item $\Delta = (\Gamma', \Cassoc{\tilde{r}}{\tilde{\tR}})$ for some $\Gamma'$
\item $\Gamma \ar{\cancom{u}{a}{p}} \Gamma'$ %for some owners $\ownfnt u, \ownfnt p$
\item $ M \structeq \Cnew{\tilde{r}}{\tilde{\tR}} (M' \Cpar \Cloc{a!\pc{v}. Q}{u})$
\item $ N \structeq  (M' \Cpar \Cloc{ Q}{u})$
\item  $\Theta \with M \newar{\cancom{u}{(\tilde{r}:\tilde{\tR})\alpha}{p'}} 
        \Theta',  \Cassoc{\tilde{r}}{\tilde{\tR}} \with N$ 
whenever
$\Theta \ar{\cancom{u}{a}{p'}} \Theta'$, for any owner $\ownfnt p'$.
\end{enumerate}
\end{lem}
\begin{proof}
  By induction on the derivation of 
   $\Gamma \with M \newar{\cancom{u}{(\tilde{r}:\tilde{\tR})a!v}{p}} \Delta \with N$.
\end{proof}

\begin{lem}[Deriv-input]\label{lemma:derin1}
  Suppose $\Gamma \with M \newar{\cancom{u}{(\tilde{r}:\tilde{\tR})a?v}{p}} \Delta \with N$.
Then
\begin{enumerate}[\em(i)]
\item $\Delta = (\Gamma', \Cassoc{\tilde{r}}{\tilde{\tR}})$ for some $\Gamma'$
\item $\Gamma \ar{\cancom{u}{a}{p}} \Gamma'$ %for some owners $\ownfnt u, \ownfnt p$
\item $M \structeq \Cnew{\tilde{c}}{{\tC}} (\Cloc{a?\pa{x}. T}{p} \Cpar M')$
\item $N \structeq \Cnew{\tilde{c}}{{\tC}} (\Cloc{T \subst{x}{v}}{p} \Cpar M')$
\item  $\Theta \with M \newar{\cancom{u'}{(\tilde{r}:\tilde{\tR'})\alpha}{p}} 
        \Theta',  \Cassoc{\tilde{r}}{\tilde{\tR'}} \with N$ 
whenever
$\Theta \ar{\cancom{u'}{a}{p}} \Theta'$, for any owner $\ownfnt u'$, 
and types $(\tilde{\tR'})$. 
\end{enumerate}
\end{lem}
\begin{proof}
  Again a straightforward induction on the derivation
$\Gamma \with M \newar{\cancom{u}{(\tilde{r}:\tilde{\tR})a?v}{p}} \Delta \with N$.
Note that in part ({\emph v}) arbitrary types $(\tilde{\tR'})$ can be used because there is no
restriction on the type $\tR$  in the second part of the rule \Rlts{in} in 
Figure~\ref{fig:lts}. 
\end{proof}

\begin{thm}\label{prop:red.tau}
$\Gamma \with M \smalleval \Delta \with N$ if and only if
        $\Gamma \with M \newar{\tau} \Delta \with N'$ for some 
   $N'$ such that $N \structeq N'$.
 \begin{proof}[(Outline)]

   First we need to show  the auxiliary result that structural
   equivalence is preserved by actions. That is $\Gamma \with M
   \newar{\lambda} \Delta \with M'$ and $M \structeq N$ implies $\Gamma \with N
   \newar{\lambda} \Delta \with N'$ for some $N'$ such that $M' \structeq N'$; this
   is proved by induction on the proof of the fact that $M \structeq
   N$ from the rules in Figure~\ref{fig:structeq}. Then a
   straightforward proof by induction on the derivation of $\Gamma
   \with M \smalleval \Delta \with N$ from the rules in
   Figure~\ref{fig:reductions} will show that this implies $\Gamma
   \with M \newar{\tau} \Delta \with N'$ with $N \structeq N'$; the
   auxiliary result is required when considering the rule
   $\Rred{struct}$.

   To prove the converse we also employ the two previous lemmas, giving
   the structure of input and output actions. Suppose $\Gamma \with M
   \newar{\tau} \Delta \with N$; we prove by rule induction that
   $\Gamma \with M \smalleval \Delta \with N$. The only non-trivial case is
   when this judgement is inferred using the rule \Rlts{comm}, or its dual. So 
without loss of generality we know 
   \begin{enumerate}[$\bullet$]
   \item $M = M_1 \Cpar M_2$
   \item $N = \Cnew{ {\tilde r}}{{\tilde \tR}   }(N_1 \Cpar N_2) $
   \item $\Gamma \with M_1 \newar{\cancom{u}{({\tilde{r}}:{\tilde{\tR}}) a?v}{p}} 
         \Delta,\Cassoc{\tilde{r}}{\tilde{\tR}} \with N_1$
   \item $\Gamma \with M_2 \newar{\cancom{u}{({\tilde{r}}:{\tilde{\tR}}) a!v}{p}} 
         \Delta,\Cassoc{\tilde{r}}{\tilde{\tR}} \with N_2$
   \end{enumerate}
The previous two lemmas can now be applied to obtain the structure of $M_1,\,M_2, N_1$ and
$N_2$, up to structural equivalence; by rearranging $M_1 \Cpar M_2$, again using the structural
equivalence rules, an application of \Rred{comm} followed by one of \Rred{struct} gives the
required 
 $\Gamma \with M \smalleval \Delta \with N$.
  \end{proof}
\end{thm}

\subsection{A proof methodology for  \picost}\label{sec:examples}

The operational semantics  given in the previous subsection 
can be used in a straightforward way to obtain a wLTS for \picost 
configurations. It suffices to attach a weight to the actions, which
can be done in a systematic manner: we write 
\begin{align*}
  (\Gamma \with M) \ar{\mu}_w (\Delta \with N)
\end{align*}
whenever 
\begin{enumerate}[$\bullet$]
\item $(\Gamma \with M) \newar{\mu} (\Delta \with N)$ can be deduced
from the rules in Figure~\ref{fig:lts} and Figure~\ref{fig:ltsB}

\item $w = (\Delta^\record - \Gamma^\record)$
\end{enumerate}
Note that the weight associated with an action is ultimately determined by the manner
in which expenditure is recorded in the cost environments; this may reflect the cost of
providing the resource in question, as in Example~\ref{ex:lib}, the profit to be gained
by a particular owner in the use of the resource, as in Example~\ref{ex:publishing}, or 
combinations of such concerns.  

We can now apply Definition~\ref{def:amort} to this wLTS  to obtain 
a family of preorders
\begin{align}\label{eq:amort.picost}
  (\Gamma \with M) \amort{n} (\Delta \with N)
\end{align}
between \picost configurations. However we must be somewhat careful here, as some of
the actions used involve bound names; but by a systematic application of \emph{Barendregt's
convention}, mentioned on page~\pageref{barendregt}, confusions between these and free names
can be avoided.

As is well-known, the relations (\ref{eq:amort.picost}) come equipped with a powerful co-inductive
proof methodology. In order to prove 
\begin{math}
   (\Gamma \with M) \amort{k} (\Delta \with N)
\end{math}
for a particular $k$ it is sufficient to exhibit a family of relations 
$\setof{\calR^n}{n \in \nats}$ which satisfy the transfer properties of 
Definition~\ref{def:amort}, such that $\calR^k$ contains the pair 
$(\Gamma \with M,\Delta \with N)$. 
In the remainder of this section we apply this proof methodology to the 
examples in Section~\ref{sec:lang}. This allows us to   now reason about the behaviour of
systems, how they  interact with other systems, rather than reason simply about their computation
runs.

\begin{exa}[Running a library, revisited]\label{ex:libagain}
  Refering to the definitions in Example~\ref{ex:lib}, by exhibiting a witness weighted bisimulation it is
possible to show
\begin{align*}
  ( \Gamma_{\central} \with \Cloc{\Reader}{\Cpublic})   \amort{0} (\Gamma_{\local} \with \Cloc{\Reader}{\Cpublic})
\end{align*}
This is despite the fact that the local use of the service $\reqR$ is
more expensive than the central use; this is compensated for by the
fact that both $\goLib$ and $\goHome$ are less expensive locally. It
is also worth noting that although the \emph{use} of resources in both
$\Gamma_{\central}$ and $\Gamma_{\local}$ is free, in the generated
wLTS the output actions actually have non-zero weights associated with
them. For example, a typical run in this wLTS from $ (
\Gamma_{\central} \with \Cloc{\Reader}{\Cpublic}) $ takes the form
\begin{align*}
    ( \Gamma_{\central} \with \Cloc{\Reader}{\Cpublic}) \;
 \ar{\goLib?n}_{5} \ldots  \ar{(r)\reqR!(r,n)}_1 \ldots \ar{\goHome!b}_5 \ldots
\end{align*}
whereas the corresponding local run is 
\begin{align*}
    ( \Gamma_{\local} \with \Cloc{\Reader}{\Cpublic})  \;
 \ar{\goLib?n}_{1} \ldots  \ar{(r)\reqR!(r,n)}_3 \ldots \ar{\goHome!b}_1 \ldots
\end{align*}

To compare the efficiency of the library service itself we consider the following definitions
\begin{align*}
 \Lib_\local       &\Leftarrow  \Cnew{\reqS}{\tR^l_s}(\Cloc{\Library \Cpar \Store}{\Clib}) \\
 \Lib_\central       &\Leftarrow \Cnew{\reqS}{\tR^c_s}(\Cloc{\Library \Cpar \Store}{\Clib})  
\end{align*}
where, as explained in Example~\ref{ex:lib}, 
$\tR_l^s,\;\tR_c^s$, are the types $\pair(0,5),\, \pair(0,1)$ respectively; here the interaction
between the library and the store has been internalised, with  types  reflecting the relative cost of
local and central access.  Both these configurations simply \emph{provide} the service 
$\reqR$, and viewed in isolation the local service is not more efficient than the central one; no matter what
$n$ we choose, we have 
\begin{align}\label{ex:a}
  ( \Gamma_{\central} \with \Lib_\central)   \not{\amort{n}} (\Gamma_{\local} \with \Lib_\local)   
\end{align}
However if we combine the library service with the reader then the overall systems is locally more efficient
than the centralised one:
\begin{align}\label{ex:b}
  ( \Gamma_{\central} \with \Sys_\central)   \amort{2} (\Gamma_{\local} \with \Sys_\local)   
\end{align}
where
\begin{align*}
  \Sys_\local &\Leftarrow \Cnew{\reqR}{\tR^l_r} (\Cloc{\Reader}{\Cpublic} \;\Cpar\; \Lib_\local)\\
  \Sys_\central &\Leftarrow \Cnew{\reqR}{\tR^c_r} (\Cloc{\Reader}{\Cpublic} \;\Cpar\; \Lib_\central)
\end{align*}
We should point out that in (\ref{ex:a}) and (\ref{ex:b}) we have used the full cost environments
$\Gamma_\local,\;\Gamma_\central$, despite the fact that some of the resources have been restricted in
the systems; this is simply in order to avoid the definition of even more environments. 

As an example of how such statements can be proved see the Section~\ref{app:library} in the appendix for a witness bisimulation which
establishes (\ref{ex:b}).  
\boxHere
\end{exa}

\section{Contextual characterisation}\label{sec:cxt}

In the previous section we have demonstrated that the preorders
$\amort{n}$ provide a useful co-inductive methodology for comparing the
behaviour of processes, relative to resource costs.  In this section
we critically review its formulation, revealing some significant
inadequacies, and offer a revised version where these are addressed.

Informally we would expect at least the following two properties of
a proof methodology:
\begin{enumerate}[(a)]
\item 
It should support \emph{compositional reasoning}, whereby 
the analysis of process behaviour can be carried out structurally.

\item
\emph{Soundness:}
Any relationship established between the behaviour of processes using the proof methodology should be
justifiable in some independent manner.
\end{enumerate}
Further we could hope for:
\begin{enumerate}[(c)]
\item 
\emph{Completeness:} 
any pair of processes which are intuitively behaviourally related, should be provably related
using our methodology. 
\end{enumerate}

Relative to our language \picost the first criteria, (a), is straightforward 
to formalise, as a property of the preorders $\amort{n}$.

\begin{defi}[Compositional]\label{def:comp}
  A relation $\calR$ over \picost configurations is said to be 
\emph{compositional} whenever
$(\Gamma \with M) \;\calR\; (\Delta \with N)$  implies 
 \begin{enumerate}[(i)]
 \item 

$(\Gamma \with M \Cpar O) \;\calR^m\; (\Delta \with N \Cpar O)$, provided 
$(\Gamma \with M \Cpar O)$ and $(\Delta \with N \Cpar O)$ are configurations

\item $(\Gamma, \Cassoc{r}{\tR} \with M) \;\calR^m\; (\Delta, \Cassoc{r}{\tR}  \with N)$.
\boxHere
\end{enumerate}

\end{defi}
We could of course  demand that  the relation $\calR$ should be preserved by all
the  operators in the language, but for the purposes of the discussion to follow it is
sufficient to concentrate on the two most important ones. 

Our first remark is that the relations $\amort{n}$ are not compositional,
and therefore our proposed proof methodology does not support compositional reasoning. 

\begin{exa}[Non-compositionality]\label{ex:noncomp}
  Let $\Gamma$ be a cost environment with two owners $\ownfnt o,\ownfnt
  p$ and two resources $a,b$.  Suppose further that $\Gamma^o({\ownfnt
    o}) = \Gamma^o({\ownfnt p}) = \infty$, while $\Gamma^u(a) =20,\, \Gamma^u(b) = 10$; the
  remaining fields in $\Gamma$ are unimportant, but to be definite let us
say that  $\Gamma^p(a) =  \Gamma^p(b) = 0$. Let $\Delta$ be
  another cost environment with the same resources, with both usage
  costs being $10$, and the same  owners, but  with the difference that
  $\Delta^o({\ownfnt o}) = 10$.  Then it is easy to check that
\begin{align*}
    \Gamma \with \Cloc{a!}{o}  &\amort{0}  \Delta \with \Cloc{a!}{o} 
\end{align*}
However one can also show that
\begin{align*}
         \Gamma \with \Cloc{a!}{o} \Cpar \Cloc{b!}{o}  &\not{\amort{0} } 
        \Delta \with \Cloc{a!}{o} \Cpar \Cloc{b!}{o} 
\end{align*}
The problem occurs when we consider the 
action
\begin{math}
   (\Gamma \with \Cloc{a!}{o} \Cpar \Cloc{b!}{o})
   \newar{\cancom{o}{b!}{p}}_{10} 
  (\Gamma_1 \with \Cloc{a!}{o} \Cpar \Cloc{\Cstop}{o}).
\end{math}
This can be matched by the action
\begin{math}
   (\Delta \with \Cloc{a!}{o} \Cpar \Cloc{b!}{o})
   \newar{\cancom{o}{b!}{p}}_{10}    
  (\Delta_1 \with \Cloc{a!}{o} \Cpar  \Cloc{\Cstop}{o})
\end{math}
but at the expense of exhausting all of ${\ownfnt o}$'s funds. $\Delta_1^o({\ownfnt o})$ is
now set to $0$ and therefore the action 
\begin{math}
   (\Gamma_1 \with \Cloc{a!}{o} \Cpar \Cloc{\Cstop}{o})
   \newar{\cancom{o}{a!}{p}}_{20} 
   (\Gamma_1 \with \Cloc{\Cstop}{o} \Cpar \Cloc{\Cstop}{o})
\end{math}
can not be matched by any action from
$  (\Delta_1 \with \Cloc{a!}{o} \Cpar  \Cloc{\Cstop}{o})$.

% Admittedly here one can show 
% \begin{math}
%          \Gamma \with \Cloc{a!}{o} \Cpar \Cloc{b!}{o}  \amort{10} 
%         \Delta \with \Cloc{a!}{o} \Cpar \Cloc{b!}{o} 
% \end{math}
% but for any $k \geq 0$ it is easy to adapt this counterexample so that it works for
% $\amort{k}$, by increasing the cost of the resource $b$.\boxHere
\end{exa}

The other criteria, (b) and (c) above, are more difficult to formalise. But
even in the absence of a precise formalisation we can also show that
our proof methodology  runs into difficulties with them, by
considering a proposed touchstone family of preorders $\behav{n}, n \geq 0$, which incorporate
some intuitive properties which we would expect.
First an easy
example, essentially taken from \cite{pityping}.
\begin{exa}[Problem with output types]\label{ex:types}
  Consider the two configurations $\calC$ and $\calD$, denoted by
  \begin{align*}
\Gamma \with \Cnew{r}{\tR_1} (\Cloc{a!\pc{r}.\Cstop}{o}), \qquad&
\Gamma \with \Cnew{r}{\tR_2} (\Cloc{a!\pc{r}.\Cstop}{o}) 
  \end{align*}
 respectively, where the types $\tR_1,\;\tR_2$ are different, and $\Gamma$ has
  sufficient resources for $a$ to be exercised; that is $\Gamma
  \ar{\cancom{o}{a}{p}} \Gamma'$ for some owner $\ownfnt p$ and some
  $\Gamma'$.

Then it is easy to see that 
\begin{math}
 \calC
\not{\amort{k}}
\calD
\end{math}
for any $k$ because the only actions which the configurations can perform are different;
they are labelled  $\cancom{p}{(\Cassoc{r}{\tR_1})a!r}{o}$ and  
$\cancom{p}{(\Cassoc{r}{\tR_2})a!r}{o}$ respectively. 

However it is difficult to envisage any  context in which these two
configurations can be distinguished; for any reasonable definition of
the touchstone relations we would expect $\calC \behav{k} \calD$ to be
true.  Thus our proof methodology will not be \emph{complete}. \boxHere
\end{exa}
Our next example focuses on some of the novel features of \picost.
\begin{exa}[Problem with owner identification]\label{ex:own}
  Let $\calC,\; \calD$ denote the configurations
  \begin{align*}
    \Gamma \with \Cloc{a!}{o_1}, \qquad&
    \Gamma \with \Cloc{a!}{o_2}
  \end{align*}
respectively, where $\ownfnt o_1,\; \ownfnt o_2$ are two different owners, 
and $\Gamma^o(\ownfnt o_1) = \Gamma^o(\ownfnt o_2)$. 

Here again we would expect $\calC \behav{k} \calD$ to be true because
there is no mechanism in \picost which would enable an observer to discover
who was funding the use of the resource $a$.  However assuming some owner $\ownfnt p$
has sufficient funds in $\Gamma$ to provide the resource $a$, we have
\begin{math}
  \calC \not{\amort{0}} \calD
\end{math}
again because the configurations perform different actions, labelled
$\cancom{o_1}{a!}{p}$ and $\cancom{o_2}{a!}{p}$ respectively. \boxHere
\end{exa}

\subsection{Behavioural preorders}

In order to address the inadequacies with our proof methodology let us
first give one possible formalisation of the touchstone family of
behavioural preorders which we have been refering to as 
$\behav{n}, n \geq 0$; we adapt the theory of
\emph{reduction barbed congruences}, \cite{ht92,pibook,pityping} to
\picost, often refered to informally as \emph{contextual equivalences}. 
For simplicity we assume that resource charging is always
standard, and that the only values used are channel/resource names.

We first need to introduce into the reduction semantics some record of the
costs being expended. Let us write 
\begin{math}
  \Gamma \with M \smalleval_c  \Delta \with N
\end{math}
whenever 
\begin{math}
   \Gamma \with M \smalleval  \Delta \with N
\end{math}
can be deduced from the reduction rules, in Figure~\ref{fig:reductions}, and 
$(\Delta^\record - \Gamma^\record)  = c$. This is generalised in the obvious manner
to 
\begin{math}
    \Gamma \with M \smalleval^*_d  \Delta \with N
\end{math}
by the accumulation of costs. 

\begin{defi}
[Cost improving] We say that the family 
of relations
$\setof{\calR^n}{ n \in \nats} $ over configurations
is \emph{cost improving} whenever 
$\calC \;\calR^m\; \calD$ for any $m$, then
\begin{enumerate}[(i)]
\item $\calC \smalleval_c \calC'$ implies $\calD \smalleval^*_d \calD'$ such that 
      $\calC' \calR^{(m+c-d)} \calD'$

\item conversely,  $\calD \smalleval_d \calD'$ implies $\calC \smalleval^*_c \calC'$ such that 
      $\calC' \calR^{(m+c-d)} \calD'$.
\boxHere
\end{enumerate}
\end{defi}
\noindent
This is a natural generalisation of the notion of \emph{reduction closure}
or \emph{reduction bisimulation} from LTSs to weighted LTSs; for a justification
of its use in defining behavioural preorders see Chapter 2 of \cite{pibook}.

\begin{defi}
[Observations] 
Let us write $(\Gamma \with M) \Downarrow a?$ whenever
$(\Gamma \with M) \smalleval^* (\Delta \with N)$ where 
for some owner $\ownfnt o$
\begin{enumerate}[(i)]
\item $N \structeq \CnewNT{\tilde{c}}(\Cloc{a?\pa{x}. T}{o} \Cpar N')$, 
and $a$ does not occur in $(\tilde{c})$

\item $\Delta \ar{\cancom{u}{a}{o}} \Delta'$ for some $\ownfnt u$ and $\Delta'$.
\end{enumerate}
The predicate $(\Gamma \with M) \Downarrow a!$ is defined in an
analogous manner.  Note that here the owner $\ownfnt o$ has to be able to pay
the appropriate costs for the barb. 

Then we say that the family of relations
$\setof{\calR^n}{ n \in \nats} $ over configurations \emph{preserves
  observations} whenever, for any $n$, $\calC_1 \calR^n \calC_2 $ 
$\calC_1 \Downarrow o$ if and only if $\calC_2 \Downarrow o$.
\boxHere
\end{defi}
\noindent
 Note
that unlike \cite{picost} we do not record the cost of making
observations; nor do we observe the owner responsible for the observation.
This means that our notion of barb is more elementary.

Example~\ref{ex:noncomp} demonstrates that demanding a behavioural
preorder to be compositional, in particular that it be preserved by
arbitrary parallel contexts, is very problematic as intuitively it
gives observers or external users of a system access to all the funds available
to owners of the system. Here we address this issue by defining a relativised
version of compositionality, relativised to the set of owners whose funds are available 
to external users.

\begin{defi}[$\ownO$-contextual]\label{def:contextual}
Let $\ownO$ be a subset of the owners $\Own$. A relation 
$\calR$ over \picost configurations is said to be 
\emph{$\ownO$-contextual} whenever
$(\Gamma \with M) \;\calR\; (\Delta \with N)$  implies 
\begin{enumerate}[(i)]
\item 
$(\Gamma \with M \Cpar \Cloc{P}{o} ) \;\calR\; (\Delta \with N \Cpar \Cloc{P}{o})$ for every 
${\ownfnt o} \in \ownO$ , provided 
$(\Gamma \with M \Cpar \Cloc{P}{o} )$ and $(\Delta \with N \Cpar \Cloc{P}{o})$ are configurations. 

\item $(\Gamma, \Cassoc{r}{\tR} \with M) \;\calR\; (\Delta, \Cassoc{r}{\tR}  \with N)$.
\boxHere
\end{enumerate}
\end{defi}

\noindent
Combining these three properties we obtain:
\begin{defi}[The contextual improvement preorder]\label{def:cxtequiv}
  Let $\setof{\Ocxtequiv{n}}{ n \in \nats}$ be the largest family (point-wise) of 
   $\ownO$-contextual relations
  over configurations which preserves observations, and is  cost
  improving.\boxHere
\end{defi}
The idea here is that we only consider the behaviour of systems
relative to contexts in which observers, or users of the systems, can
use code running under the financial authority of the owners in
$\ownO$. At one extreme we can take $\ownO$ to be the entire set of
owners $\Own$ and then observers have access to all owners, and
their funds; this gives Compositionality, as expressed in Definition~\ref{def:comp}. 
The other extreme is when observers have access to
none of the owners users in the systems under observation; in this
case the observers have to provide their own funds, to support
observations.

We now set ourselves the task of modifying the proof methodology of
Section~\ref{sec:examples} so that the informal properties (a), (b), and (c) are
enforced, relative to the touchstone preorders $\Ocxtequiv{n}$.
%, and in particular relative
%to the specialised case $\cxtequiv{n}$.   
First note that Example~\ref{ex:own} and
Example~\ref{ex:types} still apply when the informal relations
$\behav{n}$ are instantiated by the formal $\Ocxtequiv{n}$. But the problems 
presented in 
Example~\ref{ex:noncomp} depend on the choice of observers $\ownO$:

\begin{exa}[Unsoundness]\label{ex:unsound}
  Let $\Gamma,\; \Delta$ be as defined in Example~\ref{ex:noncomp}. Then we have already argued that
\begin{math}
    \Gamma \with \Cloc{a!}{\ownfnt o}  \amort{0}  \Delta \with \Cloc{a!}{o}.
\end{math}
Here we argue that 
\begin{math}
    \Gamma \with \Cloc{a!}{o}  \not{\Ocxtequiv{0}} \Delta \with \Cloc{a!}{o}.
\end{math}
whenever $\ownfnt o  \in \ownO$. 
For otherwise, this would imply
\begin{align*}
     \Gamma \with \Cloc{a!}{o} \Cpar \Cloc{P}{o} &\Ocxtequiv{0}  \Delta \with \Cloc{a!}{o} \Cpar \Cloc{P}{o}
\end{align*}
for any process $P$  which ensures that the configurations are still well-formed.  

However for a contradiction take $P$ to be 
\begin{math}
  a?. (b! \Cpar b?. \omega!)
\end{math}
where $\omega$ is some cost-free fresh channel. 
Then we can make the observation $\omega!$ on the left hand configuration but not on the
right hand one. 
\boxHere
\end{exa}

This example shows that in general $\ownfnt O$-observers  can deplete the resources of any owner
in $\ownO$, which is important if those owners have only finite funds. A significant  consequence 
is given in the next proposition, which limits the applicability of this behavioural preorder for
arbitrary $\ownO$. 
\begin{prop}\label{prop:increase}
  If $(\Gamma \with M) \Ocxtequiv{n} (\Delta \with N)$ for any $n$, then
     $\Gamma^o(\ownfnt o) =  \Delta^o(\ownfnt o)$ for every $\ownfnt o$ in $\ownO$. 
\end{prop}
\begin{proof}
  Suppose $(\Gamma \with M) \Ocxtequiv{n} (\Delta \with N)$  for
  some $n$, with $\ownfnt o$ an owner in $\ownO$.   We prove that $k \leq \Gamma^o({\ownfnt o})$ 
if and only if  $k \leq\Delta^o(\ownfnt o)$.

  Consider the process $O = \Cloc{\Cnew{r}{\tR} r! \Cpar
    r?. \omega!\pc{}}{o}$, where $\omega$ is a fresh cost-free channel, where
  $\tR$ is the resource type $\pair(k,0)$; so $r$ costs $k$ to use but
  is free to provide.  Then by compositionality we know
\begin{equation*}
  \Gamma, \omega:\tE \with M \Cpar O \;\Ocxtequiv{n}\; \Delta, \omega:\tE \with N \Cpar O
\end{equation*}
where $\tE$ denotes the trivial type $\pair(0,0)$. 

If   $k \leq \Gamma^o({\ownfnt o})$, we have 
\begin{math}
   \Gamma, \omega:\tE \with M \Cpar O  \Downarrow \omega!
\end{math}
and therefore, by the preservation of observations, 
\begin{math}
   \Delta, \omega:\tE \with N \Cpar O  \Downarrow \omega!.
\end{math}
But this is only possible if $k \leq \Delta^o({\ownfnt o})$. 

The converse argument is similar. 
\end{proof}
In effect this means that the behavioural preorders $\Ocxtequiv{n}$
can not be used to differentiate between configurations in which
owners from $\ownO$ accrue different levels of funds; a typical case
in point occurs with the systems in Example~\ref{ex:kickback}. For this
reason we are primarily interested in the extreme case, when the observers 
have no access to the funds of the owners in the systems under investigation. 
Let us introduce some special notation for these
situations.

Let $\sobs$ denote some arbitrary owner, intuitively taken to be external to the 
systems under observation.
For an arbitrary cost environment $\Gamma$ we use $\Gammaobs$ to denote the 
extended cost environment obtained by adding $\sobs$ to the domain of $\Gamma^o$ 
and setting $\Gamma^o(\sobs)$ to be $\infty$; in particular $\Gammaobs$ is only defined
whenever $\sobs$ is new to the domain of $\Gamma^o$.   Finally we use the notation
\begin{align*}
  \Gamma \with M  \Ecxtequiv{n} \Delta \with N
\end{align*}
as an abbreviation for
\begin{align*}
    \Gammaobs \with M   \Obscxtequiv{\sset{\sobs}}{n} \Deltaobs \with N
\end{align*}
Here the observer has no access to the owners' resources used in the configurations $\calC,\;\calD$ but has
an infinite amount of resources with which to run experiments. 
\leaveout{
The importance of these external owners is emphasised
by the following results:
\begin{prop}\label{prop:own1}
\qquad
\begin{enumerate}[\em(a)]
\item

  Let $\ownO_1 = \ownO \cup \sset{\sobs_1}$, where $\ownO_2 = \ownO \cup \sset{\sobs_1,\sobs_2}$ and both
  $\sobs_i$ are fresh. 
Then $\calC  \Obscxtequiv{\ownO_1}{n} \calD$ if and only if $\calC \Obscxtequiv{\ownO_2}{n}  \calD$.

\item

Suppose $\ownO_1 = \ownO_2 \cup \sset{\ownfnt o}\;  \ownO_2 = \ownO \cup \sset{o}$, where 
again $\sobs$ is fresh, and 
suppose further that
$\Gamma^o(\ownfnt o) = \Delta^o(\ownfnt o) = \infty$. Then 
$\Gamma \with M  \Obscxtequiv{\ownO_1}{n} \Delta \with N$ if and only if
$\Gamma \with M  \Ecxtequiv{n} \Delta \with N$.
\end{enumerate}
\end{prop}
\begin{proof}
  These statements are difficult to prove directly from the definitions, but are straightforward
consequences of the full-abstraction result, Theorem~\ref{thm:fa}. 
\end{proof}
The first result states informally that there is no need to have two external owners, while the
second means that in the presence of an external owner there is no need for the observer to have
access to any other owner with indefinite funds. These, taken together with Proposition~\ref{prop:increase},
means that there is limited use in considering observers with access to owners other than a single 
external one. For this reason we are mainly interested in the behavioural preorders 
$\setof{\Ecxtequiv{n}}{ n \geq 0}$. 
}

Our revised proof methodology is based on endowing \picost with the
structure of a different, more abstract, wLTS, which takes into account 
the set of owners whose funds are available to observers, and employing
Definition~\ref{def:amort} to obtain a more abstract family of
co-inductive preorders.
In order to obtain our more abstract wLTS we forget some
of the details in the labels of the actions of the operational
semantics for \picost, given in Figure~\ref{fig:lts} and
Figure~\ref{fig:ltsB}, so that they reflect not what processes can do,
but rather what external observers with access to the funds in $\ownO$ can observe them doing. 
% First a preliminary definition. 
% \begin{defi}[Fund transfer]\label{def:fundtransfer}
%   For every $k \in N$ let $\ar{\cancom{u}{k}{p}}$ be the partial function 
% over cost environments defined by letting 
% $\Gamma \ar{\cancom{u}{k}{p}} \Delta$ whenever $\Delta$ can be obtained from 
% $\Gamma$ by transferring $k$ funds from owner $\ownfnt u$ to owner $\ownfnt p$.
% Formally it holds whenever $\Delta^o(\ownfnt u) = \Gamma^o(\ownfnt u) - k,\,
% \Delta^o(\ownfnt p) = \Gamma^o(\ownfnt p) + k$ and all other components of $\Delta$
% are inherited directly from $\Gamma$. \boxHere
% \end{defi}
This leads
to \emph{abstract labels} of the following form, ranged over by $\mu$:
\begin{enumerate}[(a)]
\item internal label $\tau$ as before

\item input label $\cancomin{u}{(\tilde{r}:\tilde{\tR})a?{v}}$ 

\item output label $\cancomout{(\tilde{r})a!{v}}{p}$

%\item external label $\ext{u}{k}{p}$
\end{enumerate}
Here only one owner is recorded in the external actions; for input we note 
the user of the resource $\ownfnt u$ while for output it is the producer $\ownfnt p$.

\begin{defi}[$\ownO$-actions]\label{def:aa}
  For each abstract label $\mu$ let the corresponding $\ownO$-action $\calC \arO{\mu}{w} \calD$ be defined by
\begin{enumerate}[(a)]

\item $(\Gamma_1 \with M) \arO{\tau}{w} (\Gamma_2 \with N)$ whenever 
  \begin{math}
    (\Gamma_1 \with M) \newar{\tau} (\Gamma_2 \with N)
  \end{math}
can be deduced from the rules, where $(\Gamma_2^\record - \Gamma_1^\record) = w$.

\item $(\Gamma_1 \with M) \arO{\cancomout{(\tilde{r})a!{b}}{p}}{w} (\Gamma_2 \with N)$
  whenever  $\ownfnt p \in \ownO$  and $(\Gamma_1 \with M) \newar{\cancom{u}{(\tilde{r}:\tilde{\tR})a!{b}}{p}}
  (\Gamma_2 \with N)$ can be deduced from the rules for some
  $(\tilde{\tR})$, and some owner $\ownfnt u$,
where $(\Gamma_2^\record - \Gamma_1^\record) = w$.

\item $(\Gamma_1 \with M) \arO{\cancomin{u}{(\tilde{r}:\tilde{\tR})a?{b}}}{w} (\Gamma_2 \with N)$
  whenever  $\ownfnt u \in \ownO$ and $(\Gamma_1 \with M) \newar{\cancomin{u}{(\tilde{r}:\tilde{\tR})a?{b}}}
  (\Gamma_2 \with N)$ can be deduced from the rules for  some owner $\ownfnt p$,
where $(\Gamma_2^\record  - \Gamma_1^\record) = w$.
\end{enumerate}
Note that in (a) the set of owners $\ownO$ plays no role, but we leave it
there for the sake of uniformity. \boxHere
\end{defi}

\noindent
This endows \picost configurations with the structure of a more
abstract wLTS, whose actions depend on the set of owners $\ownO$.  We
refer to this as the $\ownO$-wLTS and we write $\calC \Oamort{n}
\calD$ whenever there is an amortised weighted bisimulation
$\setof{\calR^n}{n \in \nats}$ in this $\ownO$-wLTS such that $\calC
\calR^n \calD$.  When $\ownO$ is the singleton set $\sset{\sobs}$
where the owner $\sobs$ is fresh, that is  external to the
configurations being compared, we abbreviate this to $\calC \Eamort{n}
\calD$.

\begin{exa}[Publishing, revisited]\label{ex:publishingagain}
Here we use the notation and definitions from Example~\ref{ex:publishing} and Example~\ref{ex:kickback}.

First we can compare the profits gained by running the publishing system in different
cost environments. As before  let $\Gamma_{327}$ represent any cost environment of the
form 
$\Gammadyn, 
\Cassoc{\news}{\tR_n},
\Cassoc{\adv}{\tR_a},
\Cassoc{\publish}{\tR_p}$, where these types are
$\pair(3,1),\pair(2,0),\pair(7,1)$ respectively, and let $\Gamma_{216}$ be the same environment
but with these types changed to $\pair(2,1),\pair(1,0),\pair(6,1)$. Then it is straightforward to 
exhibit a witness bisimulation to establish
\begin{align*}
  (\Gamma_{216} \with \Cloc{P}{p}) \Eamort{0} (\Gamma_{327} \with \Cloc{P}{p})
\end{align*}
 Recall from Example~\ref{ex:publishing} that in these cost environments we record the costs of the
actions relative to their effect on the funds of $\ownfnt p$ the publisher. So this means that 
that more profit can be gained by the publisher $\ownfnt p$ by using the cost regime underlying the environment
$\Gamma_{216}$.

To investigate the effect of implementing the \emph{kickback} we consider the two systems
\begin{align*}
  \text{PA} &\Leftarrow \Cnew{\adv}{\tR_a} (\Cloc{P}{p} \;\Cpar\;\Cloc{A}{a})\\
  \text{PA}_K &\Leftarrow \Cnew{\adv}{\tR_a} (\Cloc{P_K}{p} \;\Cpar\;\Cloc{A_K}{a})\
\end{align*}
Both these systems \emph{use} the $\news$ resource and \emph{provide} the $\publish$ resource.
Here we can show, for example, that 
\begin{align*}%\label{ex:pub.proved}
  (\Gamma_{327} \with  \text{PA}_K)  \Eamort{0} ( \Gamma_{327} \with \text{PA})
\end{align*}
provided $\Gamma_{327}^o(\ownfnt p)$ is at least $5$. See Section~\ref{sec:app.pub} of the appendix for a description of a witness
bisimulation. Again because of the way in which we have set up the accounting in the cost environments this means 
that the code $\text{PA}_K$ is more profitable for the publisher 
than $\text{PA}$. 
\boxHere 
\end{exa}

The abstract $\ownO$-wLTS has precisely enough information about
actions to characterise the touchstone contextual behavioural
preorder, at least in the extreme case of $\ownO = \sset{\sobs}$.
\begin{thm}[Full-abstraction, external case]\label{thm:externalfa}
   For every $n \in \nats$,  $(\Gamma \with M) \Ecxtequiv{n} (\Delta \with N)$ if and only if 
  $(\Gamma \with M) \Eamort{n} (\Delta \with N)$. 
\end{thm}
\begin{proof}
  This will follow from the more general full-abstraction result, given in Theorem~\ref{thm:fa}.
\end{proof}
Unfortunately this result is not true for an arbitrary set of external owners $\ownO$. 
Example~\ref{ex:unsound} can be used to show that the $\ownO$-wLTS has not taken into account the fact that 
observers have access to the funds of arbitrary owners in $\ownO$.

\begin{exa}\label{ex:unsound2}
  We use the notation from Example~\ref{ex:unsound}, which in turn is inherited from
Example~\ref{ex:noncomp}. Let $\ownO$ be a set of owners which includes $\ownfnt o$ and the fresh $\sobs$.
Then it is easy to check that
\begin{math}
    \Gamma \with \Cloc{a!}{o}  \Oamort{0}  \Delta \with \Cloc{a!}{o}.
\end{math}
But we have already argued in Example~\ref{ex:unsound} that 
\begin{math}
    \Gamma \with \Cloc{a!}{o}  \not{\Ocxtequiv{0}} \Delta \with \Cloc{a!}{o}.
\end{math}
\boxHere
\end{exa}
So we have to revise the $\ownO$-wLTS to take into account the access which observers may have to funds being used by
the systems under investigation. 
\begin{defi}[Fund transfer]\label{def:fundtransfer}
  For every $k \in \nats$ let $\ar{\cancom{u}{k}{p}}$ be the partial
  function over cost environments defined by letting $\Gamma
  \ar{\cancom{u}{k}{p}} \Delta$ whenever $\Delta$ can be obtained from
  $\Gamma$ by transferring $k$ funds from owner $\ownfnt u$ to owner
  $\ownfnt p$.  Formally this partial function is only defined when 
  $\Gamma^o(\ownfnt u) \geq k$, in which case  
  $\Delta^o(\ownfnt u) = \Gamma^o(\ownfnt u) - k,\, \Delta^o(\ownfnt
  p) = \Gamma^o(\ownfnt p) + k$, when $\ownfnt p \not = \ownfnt u$ and
  all other components of $\Delta$ are inherited directly from
  $\Gamma$; when $\ownfnt p = \ownfnt u$ the operation leaves $\Delta$
  unchanged.  This leads to a new action over configurations, with a
  new abstract label $\ext{u}{k}{p}$: we let
 $$
(\Gamma_1 \with M) \arO{ \ext{u}{k}{p}}{w} (\Gamma_2 \with M)
$$
  whenever $\Gamma_1  \ar{\cancom{u}{k}{p}} \Gamma_2$, and $\ownfnt u, \ownfnt p$ are owners 
in $\ownO$, where 
 $w = (\Gamma_2^\record - \Gamma_1^\record)$.
\boxHere
\end{defi}
This gives rise to yet another LTS whose states are \picost
configurations, which we refer to as $\ownO$-awLTS, which induces
another bisimulation preorder. But we also need to take
Proposition~\ref{prop:increase} into account.
\begin{defi}[Abstract weighted bisimulation preorder]\label{def:aamort}
  A family of relations over \picost configurations $\setof{\calR^n}{n
    \in \nats}$ is said to be a \emph{$\ownO$-abstract amortised weighted
    bisimulation} whenever
  \begin{enumerate}[(i)]
  \item $\Gamma \with M  \;\calR^n\; \Delta \with M' $ implies $\Gamma^o(\ownfnt o) =  \Delta^o(\ownfnt o)$ for
every $\ownfnt o$ in $\ownO$

  \item $\setof{\calR^n}{n \in \nats}$ is an amortised weighted
    bisimulation in $\ownO$-awLTS. 
  \end{enumerate}
We write $\calC \Oaamort{n} \calD$ to denote the maximal family of such relations. \boxHere
\end{defi}
Note that these relations $\setof{\Oaamort{n}}{n\in \nats}$ actually coincide with $\setof{\Eamort{n}}{n\in \nats}$ when 
$\ownO$ is the singleton external observer $\sset{\sobs}$; this follows because the extra fund transfer actions
have no effect:
\begin{math}
  (\Gammaobs_1 \with M) \arGen{ \ext{u}{k}{p}}{w}{ \sset{\sobs}} (\Gammaobs_2 \with M)
\end{math}
if and only if  $\Gammaobs_1 = \Gammaobs_2$.

It  also coincides with the preorders used in Section~\ref{sec:examples}, under certain conditions.

\begin{prop}
Let $\ownO$ be the set of owners used in the two configurations
$\Gamma$ and $\Delta$ and suppose that  all owners in $\ownO$ have
  indefinite funds; that is $\Gamma(o) = \Delta(o) = \infty$ 
for every owner $\ownfnt o \in \ownO$. Then
  $\Gamma \with M \amort{n} \Delta \with N$ implies
$\Gamma \with M \Oaamort{n} \Delta \with N$. 
\end{prop}
\begin{proof}
  Straightforward.  When funds are unlimited the constraint (i) in
  Definition~\ref{def:aamort} is vacuous, as is the requirement to
  match the fund actions labelled $\ext{u}{k}{p}$.  The result now
  follows because every concrete action in the wLTS used in Section~\ref{sec:examples}
   is automatically also an abstract action in $\ownO$-awLTS.
\end{proof}
It follows  that the work of Section~\ref{sec:examples}
has not been in vain; the proofs in the examples can be taken to be 
about the more abstract  preorders $\Oaamort{n}$. 

The remainder of this section is devoted to showing that, subject to a
minor restriction, the co-inductive proof methodology based on
$\setof{\Oaamort{n}}{ n \in \nats}$ satisfies the informal criteria (a),
(b), and (c) set out at the begining of this section. It has certain
advantages over that used in Section~\ref{sec:examples}; in matching
input and output moves the principles involved do not have to match
up exactly.  However in the general case it also has a disadvantage with
cost environments in which certain owners have finite funds. If the
observer has access to such owners then is necessary to establish that
the proposed relations between configurations are invariant under the
transfer of funds between them. Of course in the particular case of a
purely external observer, where $\ownO$ is taken to be $\sset{\sobs}$,
which is possibly the most interesting case, then this requirement is
vacuous.

\begin{defi}[Simple types] \label{def:simpletypes}
The type $\tR = \pair(k_u,k_p)$ is \emph{simple} whenever $k_p =0$, 
meaning that resources of type $\tR$ cost nothing to provide. 
A cost environment is called \emph{simple} whenever it can be written as
$\Gammadyn, \Cassoc{a_1}{\tR_1},\ldots \Cassoc{a_n}{\tR_n}$ where $\Gammadyn$ is
a basic environment and all $\tR_i$ are simple. 

Restricting attention to simple types we know that for every resource name
$a$  there is some $k \in \nats$ such that 
$\Gamma \ar{\cancom{u}{a}{p}} \Delta$ if and only if 
$\Gamma \ar{\cancom{u}{k}{p}} \Delta$.
%This in turn means that  whenever $(\Gamma_1 \with M) \ar{ \ext{u}{k}{p}}_w (\Gamma_2 \with M)$
%we know that $w=k$. 
\boxHere
\end{defi}

\begin{thm}[Full-abstraction]\label{thm:fa} Assuming simple cost environments, for every 
set of observers $\ownO$ and 
every $n \in \nats$,  $(\Gamma \with M) \Ocxtequiv{n} (\Delta \with N)$ if and only if 
  $(\Gamma \with M) \Oaamort{n} (\Delta \with N)$. 
\end{thm}
The proof of this result is the subject of the remainder of this section; we will
also see how the restriction to simple types can be lifted, at the expense of a
generalisation of the  fund action  from  Definition~\ref{def:fundtransfer}.

\leaveout{
Before embarking on the proof of Theorem~\ref{thm:fa} we should point out that it has
Proposition~\ref{prop:own1} 
as a direct corollary. For example, using the notation of that proposition, it is simple to
prove that 
\begin{math}
  \calC \Obsaamort{\ownO_1}{n} \calD
\end{math}
if and only if
\begin{math}
  \calC \Obsaamort{\ownO_2}{n} \calD
\end{math}
because every $\ownO_1$-action is a $\ownO_2$-action and vice-versa. 
Similar reasoning gives Theorem~\ref{thm:externalfa} as a corollary.
}

\subsection{Full abstraction}\label{sec:fa}

First let us consider criteria (a) above, Compositionality.  
In fact we now  have a parametrised version of this, $\ownO$-contextuality from
Definition~\ref{def:contextual}, which we tackle  in two steps.
First we require a lemma. 
\begin{lem}\label{lem:env.extend}\qquad
  \begin{enumerate}[\em(i)]
  \item   Suppose $\Gamma \with M \newar{\lambda} \Delta \with N$. Then 
  $\Gamma, \Cassoc{r}{\tR}  \with M \newar{\lambda} \Delta,\Cassoc{r}{\tR} \with N$.

 \item Conversely, suppose $\Gamma,\Cassoc{r}{\tR} \with M \newar{\lambda} \Delta,\Cassoc{r}{\tR} \with N$,
       where the label $\lambda$ does not describe a communication along the channel $r$. 
       Then 
       \begin{enumerate}[\em(a)]
       \item $\Gamma \with M \newar{\lambda} \Delta \with N$
       \item or the concrete action label $\lambda$ is of the form $\cancom{u}{a?r}{p}$, in which case
        $\Gamma \with M \newar{ \cancom{u}{(r:\tR)a?r}{p} } \Delta, \Cassoc{r}{\tR} \with N$.
       \end{enumerate}
  \item $\Gamma \with M  \newar{\cancom{u}{(r:\tR)a?r}{p}} \Delta, \Cassoc{r}{\tR} \with N$
   implies $\Gamma, \Cassoc{r}{\tR} \with M  \newar{\cancom{u}{a?r}{p}} \Delta, \Cassoc{r}{\tR} \with N$

  \end{enumerate}
\end{lem}
\proof
   Each statement is  proved by induction on the derivation of the judgement. 
Note that for any $a$ in the domain of $\Gamma$, 
$\Gamma \ar{\cancom{u}{a}{p}} \Delta$ if and only if 
 $\Gamma, \Cassoc{r}{\tR} \ar{\cancom{u}{a}{p}} \Delta,\Cassoc{r}{\tR}$.\qed

\begin{prop}[$\ownO$-contextual]\label{prop:a.extend}
  $ (\Gamma \with M) \Oaamort{n} (\Delta \with N)$ implies $ (\Gamma,
  \Cassoc{r} {\tR} \with M) \Oaamort{n} (\Delta, \Cassoc{r}{\tR} \with
  N)$.
\end{prop}
\proof
  Let $\setof{\calR^n}{n \in \nats}$ be the family of relations over \picost configurations
defined by letting
 $ (\Gamma, \Cassoc{r}{\tR} \with M) \calR^n (\Delta, \Cassoc{r}{\tR} \with N)$ 
whenever  
\begin{enumerate}[(i)]
\item either $ (\Gamma \with M) \Oaamort{n} (\Delta \with N)$

\item or $(\Gamma, \Cassoc{r}{\tR} \with M) \Oaamort{n} (\Delta, \Cassoc{r}{\tR} \with N)$.
\end{enumerate}
It is sufficient to show that this satisfies the conditions in
Definition~\ref{def:aamort}. Note that condition (i) of this
definition is trivial.

So suppose  $ (\Gamma, \Cassoc{r}{\tR} \with M) \calR^n (\Delta, \Cassoc{r}{\tR} \with N)$
and  $ (\Gamma, \Cassoc{r}{\tR} \with M) \arO{\mu}{v} (\Gamma',\Cassoc{r}{\tR} \with M')$ is an abstract action. 
We have to find a matching abstract move
$(\Delta, \Cassoc{r}{\tR} \with N) \darO{\hat{\mu}}{w} (\Delta', \Cassoc{r}{\tR} \with N')$.  
Let us look at the concrete action underlying this abstract action,  % from $ (\Gamma, \Cassoc{r}{\tR} \with M)$,
 $ (\Gamma, \Cassoc{r}{\tR} \with M) \newar{\lambda} (\Gamma', \Cassoc{r}{\tR} \with M')$. Since we know
  $(\Gamma \with M)$ is a configuration $\lambda$ can not describe a communication along $r$, 
and so we can apply part (2) of the previous lemma, to obtain two cases:
\begin{enumerate}[(a)]
\item $\Gamma \with M \newar{\lambda} \Gamma' \with M'$. In this case the required matching move
can be obtained using the fact that $  (\Gamma \with M) \Oaamort{n} (\Delta \with N)$, together with
an application of part (1) of Lemma~\ref{lem:env.extend}. 

\item $\lambda$ is the input action  $\cancom{u}{a?r}{p}$, and  
$\Gamma \with M \newar{ \cancom{u}{(r:\tR)a?r}{p}} \Gamma',\Cassoc{r}{\tR}  \with N$. Here we again use 
the fact that 
$  (\Gamma \with M) \Oaamort{n} (\Delta \with N)$ to find a matching weak concrete move from 
     $(\Delta \with N)$ labelled
$\cancom{u}{(r:\tR)a?r}{p'}$ for some owner $\ownfnt p'$. 
Part (3) of Lemma~\ref{lem:env.extend}
can now be used to transform this into a required matching move from 
$(\Delta, \Cassoc{r}{\tR} \with N)$. In this case the matching will be because of clause (ii) in the 
definition of the family $\calR^n$.\qed  
\end{enumerate}

\begin{thm}[$\ownO$-contextual]\label{thm:compositionProof}
  Suppose $(\Gamma \with M \Cpar \Cloc{P}{o} )$ and $(\Delta \with N \Cpar \Cloc{P}{o} )$
  are both configurations, where $\ownfnt o \in \ownO$.  Then $(\Gamma \with M) \Oaamort{k} (\Delta
  \with N)$ implies $(\Gamma \with M \Cpar \Cloc{P}{o} ) \Oaamort{k} (\Delta \with
  N \Cpar  \Cloc{P}{o})$.
\end{thm}
\begin{proof}
  We follow the standard proof structure, see Section 2.3 of
  \cite{pibook}, Proposition 6.4 of \cite{pityping}, Proposition 2.21
  of \cite{dpibook};  however the precise details are somewhat
  different.  Let $\setof{ \calR^n}{ n \in \nats}$ be the smallest family
  of relations which satisfies:
\begin{enumerate}[(i)]
\item 
$\Gamma \with M  \aamort{n}  \Delta \with N$ implies
$\Gamma \with M  \calR^n  \Delta \with N$ 

\item 
$\Gamma \with M  \calR^n  \Delta \with  N$ implies 
$(\Gamma \with M \Cpar \Cloc{P}{o} ) \calR^n (\Delta \with
  N \Cpar \Cloc{P}{o})$, whenever $\ownfnt o \in \ownO$ and 
both  $(\Gamma \with M \Cpar \Cloc{P}{o} )$ and $(\Delta \with N \Cpar \Cloc{P}{o} )$
  are  configurations

\item 
$\Gamma, r:\tR_1 \with M  \calR^n  \Delta, r:\tR_2 \with  N$ implies 
$\Gamma \with \Cnew{r}{\tR_1} M  \calR^n \Delta \with \Cnew{r}{\tR_2} N$.
\end{enumerate}
We show that this family satisfies the requirements of Definition~\ref{def:aamort}, up to structural
equivalence, from which
the result will follow.

First note that for any $n$, 
\begin{equation}
  \label{eq:extend}
 \Gamma \with M \calR^n \Delta \with N \;\;
\text{implies}\;\;   \Gamma, \Cassoc{r}{\tR} \with M \calR^n \Delta, \Cassoc{r}{\tR} \with N 
\end{equation}
This can be proved by induction on
why $\Gamma \with M \calR^n \Delta \with N$, with the base case being
provided by Proposition~\ref{prop:a.extend}.  

So suppose $\Gamma \with M
\calR^n \Delta \with N$ and $\Gamma \with M \arO{\mu}{v} \Gamma' \with
M_d$; we have to find a matching abstract  move $\Delta \with N \darO{\hat{\mu}}{w}
\Delta' \with N_d$ such that $\Gamma' \with M_d \calR^{(n+v-w)} \Delta'
\with N_d$; the symmetric requirement, of matching a move from $\Delta
\with N$ by a corresponding one from $\Gamma \with M$, is treated in
an analogous fashion.

We proceed by induction on why $\Gamma \with M \;\calR^n\; \Delta \with N$,
there being three cases, (i), (ii) and (iii) above, to consider.  In the first case
the requirement comes from 
Proposition~\ref{prop:amort.prop}. We concentrate on case (ii), where
we know 
 $M,N$ have the form $(M' \Cpar \Cloc{P}{o}),\, (N' \Cpar \Cloc{P}{o})$ respectively,
where $\ownfnt o \in \ownO$ 
and we know by induction that 
\begin{math}
  \Gamma \with M' \calR^n \Delta \with N'.
\end{math}
We now examine why $\Gamma \with M' \Cpar \Cloc{P}{o} \arO{\mu}{v} \Gamma' \with
M_d$, and to start let us assume that $\mu$ is the label
$\ext{u}{k}{p}$, where the reasoning is straightforward. 
This means, by definition, that  $M_d$ is $M \Cpar  \Cloc{P}{o}$, $\ownfnt u,\ownfnt p$ are in $\ownO$ and 
$\Gamma \ar{\cancom{u}{k}{p}} \Gamma'$, which in 
turn implies $\Gamma \with M' \arO{\ext{u}{k}{p}}{v} \Gamma \with M'$;
moreover incidently $k$ and $v$ must coincide, although this fact is not required here. 
By induction this can
be matched by an action 
\begin{math}
  \Delta \with N' \darO{\ext{u}{k}{p}}{w} \Delta' \with N'' 
\end{math}
such that 
\begin{math}
  (\Gamma \with M')  \calR^{(n+v-w)} (\Delta \with N''). 
\end{math}
This matching action can now be transformed into an action of the 
form 
\begin{math}
  \Delta \with N' \Cpar  \Cloc{P}{o} \darO{\ext{u}{k}{p}}{w} \Delta' \with N''  \Cpar  \Cloc{P}{o}
\end{math}
which is easily seen to be the required matching abstract  move.  

Having disposed of this simple case we now know that 
there is a  derivation using the rules from
Figure~\ref{fig:lts}, Figure~\ref{fig:ltsB} of the underlying action
\begin{align}\label{eq:move}
  \Gamma \with M' \Cpar   \Cloc{P}{o} \newar{\lambda} \Gamma' \with M_d,
\end{align}
where $v = (\Gamma^{'\record} - \Gamma^\record)$, and  $\lambda$ is the more
concrete version of the label $\mu$.
If $M'$ is responsible for the concrete action (\ref{eq:move}), 
then a straightforward application of the induction hypothesis
will provide the required corresponding move. Suppose instead that $\Cloc{P}{o}$ is
responsible, that is (\ref{eq:move}) takes the form
\begin{align}\label{eq:movePo}
  \Gamma \with M' \Cpar  \Cloc{P}{o}  \newar{\lambda} \Gamma' \with M' \Cpar  \Cloc{P'}{o}
\end{align}
because $\Gamma \with  \Cloc{P}{o}  \newar{\lambda} \Gamma' \with   \Cloc{P'}{o}$; here the reasoning
needs to be more involved. 
\begin{enumerate}[(a)]
\item First suppose this move is external, say an output with label
  $\lambda$ being
  $\cancom{o}{(\Cassoc{\tilde{r}}{\tilde{\tR}})a!v}{p}$ for some owner $\ownfnt p$. 
Because we are actually matching $\ownO$-actions we know that this $\ownfnt p$ is actually in $\ownO$. 

 Applying 
  Lemma~\ref{lemma:derout1} we know that $\Gamma'$ has the form
  $\Gamma'', \Cassoc{\tilde{r}}{\tilde{\tR}}$, where $\Gamma
  \ar{\cancom{u}{a}{p}} \Gamma''$.
The use of simple types means that $\Gamma^p(a) = 0$ and $\Gamma^u(a) = k$ for some
$k$, and standard resource charging implies that this $k$ is actually $v$. Thus we have
the external move
\begin{math}
  \Gamma \with M'  \arO{\ext{u}{k}{p}}{v} \Gamma'' \with M'
\end{math}
and  we know 
by induction  this move can be matched by some 
\begin{math}
  \Delta \with N'  \darO{\ext{u}{k}{p}}{w} \Delta'' \with N'
\end{math}
such that
\begin{math}
  \Gamma'' \with M'  \calR^{(n+v-w)} \Delta'' \with N'.
\end{math}
This matching move actually has the form
\begin{align}\label{eq:move2}
  \Delta \with N'  \darO{\tau}{w_1} \Delta_1 \with N'_1  \arO{\ext{u}{k}{p}}{k} \Delta_2 \with N'_1
     \darO{\tau}{w_3}  \Delta'' \with N'
\end{align}
with $w = w_1 + k + w_3$. 

An application of part (iv) of Lemma~\ref{lemma:derin1} or Lemma~\ref{lemma:derout1} gives the move
\begin{math}
  \Delta_1 \with \Cloc{P}{o}
 \newar{\cancom{u}{(\tilde{r}:\tilde{\tR})\alpha}{p}} \Delta_2, \Cassoc{\tilde{r}}{\tilde{\tR}} \with \Cloc{P'}{o}
\end{math}
which can be combined with the pre- and post- $\tau$ moves in (\ref{eq:move2}) to give
\begin{math}
  \Delta \with N' \Cpar \Cloc{P}{o}  \dar{\lambda}_w 
    \Delta'', \Cassoc{\tilde{r}}{\tilde{\tR}} \with N' \Cpar \Cloc{P'}{o}.
\end{math}
This is the required matching move since we know 
\begin{math}
  \Gamma'' \with M'  \calR^{(n+v-w)} \Delta'' \with N',
\end{math}
from which
\begin{math}
  \Gamma'', \Cassoc{\tilde{r}}{\tilde{\tR}} \with M' \Cpar  \Cloc{P'}{o} \;\calR^{(n+v-w)} \;
      \Delta'',\Cassoc{\tilde{r}}{\tilde{\tR}} \with N' \Cpar \Cloc{P'}{o}
\end{math}
follows by the remark (\ref{eq:extend}) above and the definition of the family $\setof{\calR^k}{ k \geq 0}$.

When the label $\lambda$ in the move (\ref{eq:movePo}) above is an input the argument is very  much the same but
with an application of 
Lemma~\ref{lemma:derin1} in place of Lemma~\ref{lemma:derout1}; it is therefore omitted. 

\item

Now suppose the move from $\Cloc{P}{o}$ we are examining is an internal move, taking the form
$\Gamma \with \Cloc{P}{o}  \newar{\tau} \Gamma' \with  \Cloc{P'}{o}$. Here we apply Theorem~\ref{prop:red.tau} and 
Proposition~\ref{prop:red}, which tell us that there are in principle three 
possibilities, (i), (ii)  or (iii). But an analysis of the proof will show that for processes of the form
$\Cloc{P}{o}$ case (i) is actually the only possibility. 
Here  $\Gamma'$ coincides with $\Gamma$, implying  incidently that
$v=0$. As we know $\Delta \with \Cloc{P}{o}$ is a configuration we also get
\begin{math}
  \Delta \with \Cloc{P}{o} \newar{\tau} \Delta \with \Cloc{P'}{o}
\end{math}
and therefore that 
\begin{math}
  \Delta \with N' \Cpar \Cloc{P}{o} \ar{\tau}_0 \Delta \with N' \Cpar \Cloc{P'}{o}.
\end{math}
It is easy to now check that this is the required matching move, since by definition
\begin{math}
  \Gamma \with M \Cpar \Cloc{P}{o} \calR^n  \Delta \with N \Cpar \Cloc{P'}{o}.
\end{math}

% Finally let us consider case (iii), where the internal move by $\Cloc{P}{o}$ is
% because of the use of a resource private to $P$; we know
% $\Gamma,\Cassoc{a}{\tR} \ar{\cancom{o}{a}{o}} \Gamma',
% \Cassoc{a}{\tR}$, for some (fresh) resource $a$.  Again because $\tR$
% must be a simple type this change in the cost environment  can be mimiced by $\Gamma
% \ar{\cancom{u}{k}{p}} \Gamma'$ for some $k \geq 0$. We know that
% $\Gamma \ownless \Delta$ which means that $\Delta
% \ar{\cancom{u}{k}{p}} \Delta'$ for some $\Delta'$, which in turn
% implies $\Delta,\Cassoc{a}{\tR} \ar{\cancom{u}{a}{p}} \Delta',
% \Cassoc{a}{\tR}$. Moreover, because we are assuming resource charging
% is standard we know $(\Delta'^\record - \Delta^\record) = v$.
% Another call on part (ii) of
% Proposition~\ref{prop:red} gives the move $\Delta \with O \newar{\tau}
% \Delta' \with O'$, which  again can be used to construct the required
% matching abstract move $\Delta \with N \Cpar O \ar{\tau}_v N \Cpar O'$.
\end{enumerate}

We are left with the possibility that the underlying action to be matched, (\ref{eq:move}) 
above, involves
communication and therefore takes the form
\begin{align*}
  \Gamma \with M' \Cpar \Cloc{P}{o}  \newar{\tau} \Gamma' \with \Cnew{\tilde{r}}{\tilde{\tR}} ( M'' \Cpar \Cloc{P'}{o})
\end{align*}
There are two cases, depending on whether $M'$ performs an input or an output. Let us consider the
latter, the former being similar but slightly easier. So we have
\begin{align}
  &\Gamma \with M' \newar{\lambda} \Gamma', \Cassoc{\tilde{r}}{\tilde{\tR}} \with M'' \notag
\\
  &\Gamma \with \Cloc{P}{o} \newar{\overline{\lambda}} \Gamma', \Cassoc{\tilde{r}}{\tilde{\tR}} \with \Cloc{P'}{o}  \label{eq:decompose}
\end{align}
with $\lambda, \overline{\lambda}$ taking the forms $\cancom{u}{ (\Cassoc{\tilde{r}}{\tilde{\tR}})a!v}{o},\, 
\cancom{u}{(\Cassoc{\tilde{r}}{\tilde{\tR}})a?v}{o}$ respectively, for some owner $\ownfnt u$.  
By induction the first move, or rather its abstract version,  can be matched  because $\ownfnt o$ is an owner in $\ownO$, giving 
\begin{align}\label{eq:move3}
  \Delta \with N' \newarstar{\tau}   
   \Delta_1 \with N'_1\newar{\cancom{u'}{(\tilde{r}:\tilde{\tR'}) a!v}{o}} 
   \Delta_2, \Cassoc{\tilde{r}}{\tilde{\tR'}} \with N'_2 \newarstar{\tau} 
      \Delta',  \Cassoc{\tilde{r}}{\tilde{\tR'}} \with N''
\end{align}
for some owner $\ownfnt u'$, such that 
\begin{math}
   (\Gamma', \Cassoc{\tilde{r}}{\tilde{\tR}} \with M'')  \;\;\calR^{(n+v-w)} 
   (\Delta',  \Cassoc{\tilde{r}}{\tilde{\tR'}} \with N''),
\end{math}
where $w = (\Delta^{'\record} - \Delta^{'\record})$. 
Note that the type of the extruded names, $\tilde{\tR'}$, may in general be different than the
types at which they were extruded by $M'$, and the owner $\ownfnt u'$ may also be different,
thereby a priori complicating matters when we try to combine this action with that from $\Cloc{P}{o}$, 
in (\ref{eq:decompose})
above. 

However an application of part (ii) of Lemma~\ref{lemma:derout1}, gives 
\begin{math}
   \Delta_1  \ar{\cancom{u'}{a}{o}} \Delta_2, 
\end{math}
and therefore from (\ref{eq:decompose}) and part (v) of Lemma~\ref{lemma:derin1} we get
\begin{math}
  \Delta_1 \with \Cloc{P}{o} \newar{ \cancom{u'}{(\tilde{r}:\tilde{\tR'}) a?v}{o}} 
\Delta_2, \Cassoc{\tilde{r}}{\tilde{\tR'}} \with \Cloc{P'}{o}.
\end{math}
This concrete move can now be combined with the concrete move (\ref{eq:move3}) to give the required matching
abstract move
\begin{math}
  \Delta \with N' \Cpar \Cloc{P}{o}   \dar{\tau}_w \Cnew{\tilde{r}}{\tilde{\tR'}} (N' \Cpar \Cloc{P'}{o}).
\end{math}
% Finally let us consider briefly  case (iii), when $M,N$ have the forms 
% $\Cnew{r}{\tR_1} M',\; \Cnew{r}{\tR_2}$  respectively, and we know
% by induction that 
% $\Gamma, r:\tR_1 \with M'  \calR^n  \Delta, r:\tR_2 \with  N'$.  
\end{proof}
The attentive reader will have noticed that the restriction to \emph{simple} types
was necessary in order to be able to model the use of a resource by  the observers using actions
based on the transfer function 
$\Gamma  \ar{ \ext{u}{k}{p}} \Delta $, which records the transfer of
$k$ funds, the cost of using the resource, from the user to the provider. 
If we drop the restriction  to \emph{simple types},  then the effect of using a resource is more complicated;
a certain amount will be debited to the user, while another  amount, possibly negative,
will be credited to the user. This can be accommodated by a more general transfer function
$\Gamma  \ar{ \ext{u}{(k_1,k_2)}{p}} \Delta $, leading in turn to a more general abstract arrow in
part (d) of Definition~\ref{def:aa}. With this adjustment compositionality can also be established
for arbitrary types. 

This contextual results leads in a straightforward manner to establishing the second informal criteria, (b):
\begin{thm}[Soundness]
   For every $n \in \nats$ and every set of owners $\ownO$,  $(\Gamma \with M) \Oaamort{n} (\Delta \with N)$
  implies  $(\Gamma \with M) \Ocxtequiv{n} (\Delta \with N)$.
\end{thm}
\begin{proof}(Outline)
  It is sufficient to show that the family of relations $\setof{\aamort{n}}{ n\in \nats}$ satisfies the 
 three defining properties of the family of contextual equivalences. \emph{Cost improving} follows by
 definition, at least up to structural induction, in view of Theorem~\ref{prop:red.tau}, and  the two preceding
results establish $\ownO$-contextuality. The final property, \emph{Preservation of observations}, is also
straightforward, since, for example, the ability to observe $a!$ from a configuration coincides with its ability 
to perform some output action on the resource $a$.  
\end{proof}

The final criteria (c), \emph{Completeness}, depends as usual on the ability to define contexts which
capture the effect of  each of the abstract $\ownO$-actions described in Definition~\ref{def:aa}. 
We first make this precise. 

\newcommand{\Csucc}{\cfn{succ}}
\newcommand{\Cfail}{\cfn{fail}}

We use two fresh cost-free resources, $\Csucc,\,\Cfail$  to record the success or failure of
tests, and a third $\Creq$ for housekeeping purposes.
For any $\Gamma$ we use  $\Gamma^t$ to denote  the cost environment  obtained by adding on these resources. 
Now let $\mu$ be an abstract action which uses the bound names $(\tilde{r})$. Then we say $\mu$ is definable
relative to $\ownO$
if for every finite set of names $F$ there exists a system  $T^F_\mu$ using only the owners from $\ownO$ such that
\begin{enumerate}[(i)]
%\item if $\dom{\Gamma^u} \subseteq F$ and $\Gamma \with M \arO{\mu}{w}
%  \Delta, \Cassoc{\tilde{r}}{\tilde{\tR}} \with N$ then $\Gamma^t
%  \with M \Cpar T^F_\mu \darO{\tau}{w} \Delta^t \with \Cnew{r}{\tR}
%  (\Csucc!\pc{\tilde{r}} \Cpar N)$ where $ M' \Downarrow \Csucc!$ and
%  $R \not \Downarrow \Cfail!$\medskip

\item if $\dom{\Gamma^u} \subseteq F$ and $\Gamma \with M \arO{\mu}{w}
  \Delta, \Cassoc{\tilde{r}}{\tilde{\tR}} \with N$ then \[\Gamma^t
  \with M \Cpar T^F_\mu \darO{\tau}{w} \Delta^t \with \Cnew{r}{\tR}
  (\Csucc!\pc{\tilde{r}} \Cpar N)\] where $ M' \Downarrow \Csucc!$ and
  $R \not \Downarrow \Cfail!$\medskip

%\item conversely, whenever $\dom{\Gamma^u} \subseteq F$, $\Gamma^t
%  \with M \Cpar T_\mu \darO{\tau}{w} \Delta^t \with M'$ where $ M'
%  \Downarrow \Csucc!$ and $M' \not \Downarrow \Cfail!$ implies $M'
%  \structeq \Cnew{\tilde{r}}{\tilde{\tR}} ( \Csucc!\pc{\tilde{r}}
%  \Cpar N)$, where $\Gamma \with M \darO{\mu}{w} \Delta,
%  \Cassoc{\tilde{r}}{\tilde{\tR}} \with N$
\item conversely, $\Gamma^t \with M \Cpar T_\mu \darO{\tau}{w}
  \Delta^t \with M'$ where $ M' \Downarrow \Csucc!$ and $M' \not
  \Downarrow \Cfail!$ implies $M' \structeq
  \Cnew{\tilde{r}}{\tilde{\tR}} ( \Csucc!\pc{\tilde{r}} \Cpar N)$,
  where $\Gamma \with M \darO{\mu}{w} \Delta,
  \Cassoc{\tilde{r}}{\tilde{\tR}} \with N$, whenever $\dom{\Gamma^u}
  \subseteq F$.
\end{enumerate}

\begin{thm}[Definability]\label{thm:definability}
  All input, output and external actions are definable.         
\end{thm}
\begin{proof}(Outline)
  Let us look at two examples. First suppose that $\mu$ is the label
$\ext{u}{k}{p}$ where $\ownfnt u$ and $\ownfnt p$ are both in $\ownO$; 
here $(\tilde{r})$ is empty and the set of names $F$ plays no role. 
The definition of $T^F_\mu$ uses a variation on Example~\ref{ex:fund.transfer}. We use
\begin{align*}
  \Cloc{\Cfail! \Cpar \Cnew{r}{\tR_k} \Creq!\pc{r}. r!.\Cstop}{u} \;\Cpar\;
  \Cloc{\Creq?\pa{x}.y?.\Cfail?.\Csucc!}{p}
\end{align*}
where $\tR_k$ is the type $(k,0)$.  This ensures that whenever $(\Gamma \with M \Cpar T^F_\mu)$
evolves at cost $w$ to a configuration $\calC$  such that  $\calC \Downarrow \Csucc$ but $\calC \not \Downarrow \Cfail$
then the newly generated resource $r$ must have been used by $\ownfnt u$ and provided by $\ownfnt p$.
This is
only possible if $\Gamma^t \with M$ can evolve to a configuration in which a transfer of $k$ can be made from
$\ownfnt u$ to $\ownfnt p$; that is a configuration $\Gamma^{t'} \with M'$ such that  $\Gamma^{t'} \ar{\cancom{u}{k}{p}} \Gamma^{t''}$.
This in turns implies that we must have $\Gamma \with M  \darO{\ext{u}{k}{p}}{w} \Delta \with N$ for some 
configuration $\Delta \with N$. Note the cost here is $w$ because all of the resources used by the test $T^F_\mu$
are cost-free. 

For the second example consider the abstract output action label $((r)a!r,\ownfnt p)$, where we know $\ownfnt p$ is
in $\ownO$.  Here we let $T^F_\mu$ be
\begin{align*}
  \Cloc{ \Cfail! \Cpar a?\pa{x}.\Cmatch{x \in F}{\Cstop}{\Cfail?.\Csucc!} }{p}
\end{align*}
where $x \in F$ is an abbreviation for a series of tests deciding whether or not $x$ is in the finite set 
of names $F$. Intuitively
whenever this is used in a cost environment $\Gamma$ satisfying $\dom{\Gamma^u} \subseteq F$ this test will
fail only when $x$ is instantiated by a fresh name.

Once more it is easy to say that the ability of $\Gamma^t \with M \Cpar T^F_m$ to evolve to a configuration
$\calC$ satisfying $\calC \Downarrow \Csucc$ but $\calC \not \Downarrow \Cfail$ coincides with the ability
of $\Gamma \with M$ to do a weak concrete move labelled  $\cancom{u}{(r:\tR)a!r}{p}$ for some owner $\ownfnt u$ and type
$\tR$. Moreover the cost of this weak concrete action will be exactly the same as the evolution from 
 $\Gamma^t \with M \Cpar T^F_m$, because the interactions with the test $T^F_\mu$ is free. 
\end{proof}

\begin{thm}[Completeness]
   For every $n \in N$  and every set of owners $\ownO$,  $(\Gamma \with M) \Ocxtequiv{n} (\Delta \with N)$ implies 
 $(\Gamma \with M) \Oaamort{n} (\Delta \with N)$.
\end{thm}
\begin{proof}(Outline)
  It suffices to show that the family 
  \begin{math}
    \setof{ \Ocxtequiv{n}   }{n \in \nats}
  \end{math}
satisfies the conditions  in Definition~\ref{def:aamort}.
Note that condition (i) is already established by Proposition~\ref{prop:increase}.
Now suppose
\begin{math}
  \Gamma \with M  \Ocxtequiv{n} \Delta \with N
\end{math}
and 
\begin{math}
   \Gamma \with M  \arO{\mu}{v} \Gamma' \with M'.
\end{math}
We have to find a matching move from 
\begin{math}
  \Delta \with N,
\end{math}
which is relatively straightforward because of Theorem~\ref{thm:definability}. As an example
suppose $\mu$ is the output label  $((r)a!r,\ownfnt p)$, and so $\Gamma'$ has the structure
$\Gamma'', \Cassoc{r}{\tR}$ for some $\tR$. 
Because of Compositionality we know 
\begin{math}
   \Gamma^t \with M \Cpar  T_\mu  \cxtequiv{n} \Delta^t \with N \Cpar T_\mu.
\end{math}
Using the first part of the Definability Theorem we know that, up to structural
equivalence,
\begin{displaymath}
  \Gamma^t \with M \Cpar T^F_\mu \smalleval^*_v \Gamma^{t''}  \with \Cnew{r}{\tR} (\Csucc!\pc{r} \Cpar M').
\end{displaymath}

Using the properties of the family $\setof{\Ocxtequiv{n}}{n \in \nats}$ this move must be matched by move
$$
\Delta^t \with N \Cpar T^F_\mu \smalleval^*_w  \Delta^{t''} \with N'' 
$$ 
where 
\begin{align}\label{eq:complete}
  \Gamma^{t''}  \with \Cnew{r}{\tR} (\Csucc!\pc{r} \Cpar M') \;\cxtequiv{(n+v-w)} 
\Delta^{t''} \with N''
\end{align}
Moreover we know 
$N'' \Downarrow \Csucc!$ and $N'' \not \Downarrow \Cfail!$ and so  the Definability theorem tells us that
$
N'' \structeq  \Cnew{r}{\tR'} (\Csucc!\pc{r} \Cpar N')
$ 
where 
$$
\Delta \with N  \darO{\mu}{w} \Delta'', \Cassoc{r}{\tR'} \with N'
$$ 
This would be the required matching move, if we had
\begin{align}
  \label{eq:complete1}
    \Gamma'', \Cassoc{r}{\tR}  \with M' \;\Ocxtequiv{(n+v-w)} 
\Delta'' , \Cassoc{r}{\tR} \with N'
\end{align}
whereas (\ref{eq:complete}) only gives us, up to structural equivalence,  
\begin{align}\label{eq:complete2}
  \Gamma^{t''}  \with \Cnew{r}{\tR} (\Csucc!\pc{r} \Cpar M') \;\cxtequiv{(n+v-w)} 
\Delta^{t''} \with  \Cnew{r}{\tR'} (\Csucc!\pc{r} \Cpar N')
\end{align}
However the so-called \emph{Extrusion Lemma}, see  Proposition 6.7 of \cite{pityping} and Lemma 2.38 of \cite{dpibook},
can easily be adapted to \picost, to show that the required (\ref{eq:complete1}) does indeed follow from
(\ref{eq:complete2})
\end{proof}

\section{Conclusion}

In this paper we have developed a behavioural theory based on
bisimulations for a version of the picalculus, \picost, in which
\begin{enumerate}[$\bullet$]
\item resources have costs associated with them
\item code runs under the financial responsibility of owners, or principals
\item code can only be executed if the owner responsible for it can finance
the available transactions.
\end{enumerate}
The behavioural theory gives rise to a  co-inductive
proof methodology for comparing the costed behaviour of systems.  We
have demonstrated the usefulness of the methodology by treating some
examples, and we have offered at least a preliminary justification for
the theory in terms of contextual requirements, parametrised on sets of
owners. We have provided some evidence that the most appropriate theory 
emerges when this set of observers is taken to be some single external
observer, external to the owners funding the systems being investigated. 
In particular with this particular set of observers there is no need 
to consider the extra actions $\ext{u}{k}{p}$ when establishing 
bisimulations. 

The language could be extended in many ways without unduely affecting
the underlying theory.  Perhaps the most obvious extension would be
the introduction of \emph{ownership types}, to control which owners
can use which resources; this would help in the modularisation of
systems. One could also introduce a scoping mechanism for owners,
limiting the range within systems of their financial responsibility. 
One effect of such  extensions would be that owners would play
a much more significant role in the (abstract) actions on which
bisimulations are based. Such investigations we leave for future work. 

The language could also be extended with mechanisms whereby processes
could be aware of which owners are funding which resources, and more
importantly base their behaviour on such knowledge. More ambitiously
the semantics of the language could be generalised so that behaviour is 
now dependent on some  \emph{dynamic cost model}. There is considerable 
scope here for inventing more realistic cost models, whereby for example 
costs associated with producing/consuming resources could vary according to
\emph{market dynamics}. It is likely that a probabilistic setting would be
most appropriate for developing such models.

The underlying theory of \emph{weighted bisimulations} also deserves
attention. For example it is not clear if the theory is decidable,
even for finite-state systems. More generally it would be interesting
to have techniques which would calculate the costs necessary to assign
to actions in order to ensure the equivalence of two systems.
There is already an extensive literature on \emph{weighted automata} \cite{wa}
and  decidability issues concerned with them, which may help in this regard.

\paragraph{Related work:} The research reported in the current paper
grew out of preliminary work reported in \cite{picost}. There a
language $\pi_{\text{cost}}$ was defined and also given a semantics
relative to \emph{cost environments}. But there are significant
differences.  At the language level the construct central to \picost,
$\Cloc{P}{o}$, is absent in $\pi_{\text{cost}}$; indeed in the latter
there is no representation of owners being responsible for specific
computations.  The cost environments used are also quite different; in
$\pi_{\text{cost}}$ funds are associated directly with resources,
which complicates considerably the reduction semantics as the resource
types need to be dynamic.  Here all funds are retained by owners,
which simplifies matters considerably, and this facilities the
introduction of \emph{charges} for resource usage and \emph{benefits}
for resource provision. Finally the behavioural theories are
different.  The concept of \emph{weighted bisimulation} is considerably
more flexible than the \emph{cost bisimulations} of \cite{picost}, as
the latter simply compares the relative cost of performing each
particular action.

\emph{Weighted bisimulations} are a direct generalisation of the
notion of \emph{amortised bisimulations} from \cite{astrid}; these
were originally defined for a version of CCS, \cite{ccs}, in which
only external actions have associated with them a cost. Nevertheless
we believe that our generalisation is significant, at least in that it
will make the concepts more generally applicable. However similar ideas
have a long history in the field of \emph{timed} process calculi; see
for example \cite{Tofts94}. A good survey of the use of
\emph{amortisation} for timed processes can be found in \cite{speed}.

Other resource-aware calculi have already appeared in the
literature. A typical example is the variant of \emph{mobile ambients}
\cite{ambients} from \cite{vladi} in which the resource in question is
\emph{space}, and the processes in the calculi have a \emph{bounded
  capacity} to host incoming ambients.  Another interesting example
may be found in \cite{teller:tcs04}, and related publications, which
develops a version of the picalculus in which unused
resources/channels may be garbage collected.  Of particular interest
to us is the general theory of \emph{resource-based} computation being
developed in \cite{pym}, and related publications. In future work we
hope to adapt their resource-based modal logic to \picost.

\appendix

\section{Some witness bisimulations}

\begin{figure}[t]
  
\begin{align*}
\text{Reader:}\qquad\qquad\qquad\qquad  
R_1           &\Leftarrow \goLib?\pa{\Cname}. \CnewNT{r}\;R_2(r,\Cname)\\
  R_2(r,\Cname) &\Leftarrow \reqR!\pc{r,\Cname}. R_3(r)\\
  R_3(r)        &\Leftarrow r?\pa{b}. R_4(b)\\
  R_4(b)        &\Leftarrow \goHome!\pc{b}. R_1\\
\\
%\end{align*}
%\begin{align*}
\text{Library:}\qquad\qquad\qquad\qquad    L_1   &\Leftarrow \reqR?\pa{y,z}. L_2(y,z) \\
 L_2(y,z) &\Leftarrow L_3(y,z) \oplus \CnewNT{r}   L_4(r,y,z) \\
  L_3(y,z) &\Leftarrow  y!\pc{\Cbook(z)}. L_1\\
L_4(r,y,z) &\Leftarrow \reqS!\pc{r,z}.L_5(y)\\
  L_5(y)   &  \Leftarrow r?\pa{b}. L_6(y,b)\\
  L_6(y,b) &\Leftarrow    y!\pc{b}. L_1 \\
%\end{align*}
\\
%\begin{align*}
\text{Store:} \qquad\qquad\qquad\qquad 
  S_1 &\Leftarrow \reqS?\pa{y,z}.S_2(y,z) \\
  S_2(y,z) &\Leftarrow y!\pc{\Cbook(z)}. S_1
\end{align*}

  \caption{Notation for library code}
  \label{fig:libcode}

\end{figure}

\subsection{The library}\label{app:library}

\newcommand{\BookN}{\ensuremath{\mathsf{BN}}\xspace}
\newcommand{\Books}{\ensuremath{\mathsf{BK}}\xspace}
\newcommand{\Deltadyn}{\Delta_{\scriptstyle dyn}}

Here we revisit the example on running a library, discussed in Example~\ref{ex:lib} and 
Example~\ref{ex:libagain}, and prove
\begin{align}\label{ex:app.lib}
   ( \Gamma_{\central} \with \Sys_\central)   \amort{2} (\Gamma_{\local} \with \Sys_\local)   
\end{align}
by exhibiting a witness bisimulation.  For convenience we work up to
\emph{structural equivalence} and modulo $\beta$-moves; essentially
these are moves which have no effect on the overall behaviour of
systems; see \cite{dpibook,groote} for details.  In \picost these include the actions  generated by the rules 
\Rlts{export}, \Rlts{unwind}, \Rlts{split}, \Rlts{match}, \Rlts{mismatch}.
Let us assume a set of book names \BookN, ranged over by $n$ and a
set of books \Books, ranged over by $b$. 

\begin{figure}[t]
  \begin{align*}
    N_1 &\Leftarrow \Cnew{\reqR}{\tR^c_r}(\Cloc{R_1}{\Cpublic} \;\Cpar\; 
                         \Cnew{\reqS}{\tR^c_s}(\Cloc{L_1}{\Clib}  \Cpar \Cloc{S_1}{\Clib}) ) \\
    N_2(n) &\Leftarrow  \CnewNT{\reqR,r}(\Cloc{R_2(r,n)}{\Cpublic} \;\Cpar\; 
                         \CnewNT{\reqS}(\Cloc{L_1}{\Clib}  \Cpar \Cloc{S_1}{\Clib}) ) \\
    N_3(n) &\Leftarrow  \CnewNT{\reqR,r}(\Cloc{R_3(r)}{\Cpublic} \;\Cpar\; 
                         \CnewNT{\reqS}(\Cloc{L_2(r,n)}{\Clib}  \Cpar \Cloc{S_1}{\Clib}) ) \\
    N_{41}(n) &\Leftarrow  \CnewNT{\reqR,r}(\Cloc{R_3(r)}{\Cpublic} \;\Cpar\; 
                         \CnewNT{\reqS}(\Cloc{L_3(r,n)}{\Clib}  \Cpar \Cloc{S_1}{\Clib}) ) \\
    N_{51}(b) &\Leftarrow  \CnewNT{\reqR,r}(\Cloc{R_4(b)}{\Cpublic} \;\Cpar\; 
                         \CnewNT{\reqS}(\Cloc{L_1}{\Clib}  \Cpar \Cloc{S_1}{\Clib}) ) \\
    N_{42}(n) &\Leftarrow  \CnewNT{\reqR,r,r'}(\Cloc{R_3(r)}{\Cpublic} \;\Cpar\; 
                         \CnewNT{\reqS}(\Cloc{L_4(r,r',n)}{\Clib}  \Cpar \Cloc{S_1}{\Clib}) ) \\
    N_{52}(n) &\Leftarrow  \CnewNT{\reqR,r,r'}(\Cloc{R_3(r)}{\Cpublic} \;\Cpar\; 
                         \CnewNT{\reqS}(\Cloc{L_5(r,r')}{\Clib}  \Cpar \Cloc{S_2(r',n)}{\Clib}) ) \\
    N_{53}(b) &\Leftarrow  \CnewNT{\reqR,r,r'}(\Cloc{R_3(r)}{\Cpublic} \;\Cpar\; 
                         \CnewNT{\reqS}(\Cloc{L_6(r,b')}{\Clib}  \Cpar \Cloc{S_1}{\Clib}) ) \\
  \end{align*}
  \caption{Library systems}
  \label{fig:libconfs}
\end{figure}

Let us write $\Gamma \sim \Delta$ whenever
\begin{enumerate}[(a)]
\item $\Gamma$ has the form 
  \begin{math}
     \Gammadyn,\; \Cassoc{\goLib}{\pair(0,5)},\;\Cassoc{\goHome}{\pair(0,5)},\;
    \Cassoc{\reqR}{\pair(0,1)},\;\Cassoc{\reqS}{\pair(0,1)}
  \end{math}
for some basic environment  $\Gammadyn$

\item $\Delta$ has the form 
  \begin{math}
     \Deltadyn,\; \Cassoc{\goLib}{\pair(0,1)},\;\Cassoc{\goHome}{\pair(0,1)},\;
    \Cassoc{\reqR}{\pair(0,3)},\;\Cassoc{\reqS}{\pair(0,5)}
  \end{math}
where again $\Deltadyn$ is 
some basic environment. 

\item $\dom{\Gamma^o} = \dom{\Delta^o} = \sset{ \Cpublic,\; \Clib}$, with $\Gamma^o(\alpha) = \Delta^o(\alpha) = \infty$, for
every $\alpha$ in its domain. 
\end{enumerate}
So effectively $\Gamma$ must be like $\Gamma_{\central}$ with perhaps a different 
record filed $\Gamma^{\record}$, and $\Delta$ must be like $\Gamma_{\local}$. 
Our witness bisimulation will contain pairs of
the form 
\begin{align*}
  \Gamma \with N &\leftrightarrow \Delta \with M   \qquad\qquad \text{where}\; \Gamma \sim \Delta
\end{align*}
The allowed forms of $N$ are described in Figure~\ref{fig:libconfs}, where for convenience we have omitted 
the explicit occurrence of the local types $\tR^c_r,\; \tR^c_s$ after the first line. These in turn use 
notation given in Figure~\ref{fig:libcode} for the various processes.  
The allowed forms for $M$ are identical  except for the use of the local types 
$\tR^l_r,\; \tR^l_s$ in place of $\tR^c_r,\; \tR^c_s$.

Let the family of relations over configurations $\setof{\calR^k}{k \in \nats}$ be determined by the following
constraints, where we assume in each clause that $\Gamma \sim \Delta$:
\begin{align*}
  \Gamma \with N_1 &\calR^k  \Delta \with M_1  &&   &\text{whenever } k \geq 2\\
  \Gamma \with N_2(n) &\calR^k  \Delta \with M_2(n)  &&   &\text{whenever } k \geq 6,\; n \in \BookN \\
  \Gamma \with N_i(n) &\calR^k  \Delta \with M_i(n)  &&   &\text{whenever } k \geq 4,\; n \in \BookN, i = 3,41,51,42 \\
  \Gamma \with N_i(n) &\calR^k  \Delta \with M_i(n)  &&   &\text{whenever } k \geq 4,\; n \in \BookN, i = 3,41,42 \\
  \Gamma \with N_i(b) &\calR^k  \Delta \with M_i(b)  &&   &\text{whenever } k \geq 0,\; b \in \Books, i = 51,52,53 \\
\end{align*}
It is fairly straightforward, although tedious, to prove that
$\setof{\calR^k}{k \in \nats}$ satisfies the requirements of being  a weak
bisimulation in the wLTS of Section~\ref{sec:examples}, up to structural
equivalence and $\beta$-moves. This is facilitated by the fact that
the code in each component of the pairs is identical.

Note that the configuration 
\begin{math}
  \Gamma_{\central} \with \Sys_\central
\end{math}
$\beta$-reduces to a configuration of the form $\Gamma \with N_1$ and 
\begin{math}
  (\Gamma_{\local} \with \Sys_\local) 
\end{math}
$\beta$-reduces to one of the form $\Delta \with M_1$, where $\Gamma \sim \Delta$,  
and thus (\ref{ex:app.lib}) above follows. 

\subsection{The publisher}\label{sec:app.pub}

\begin{figure}[t]

  \begin{alignat*}{2}
\text{Publisher:}& &\qquad
P_1(r_1)  &\Leftarrow   \news!\pc{r_1}. \CnewNT{r_2} P_2(r_1,r_2)\\
& & P_2(r_1,r_2) &\Leftarrow \adv!\pc{r_2}. P_3(r_1,r_2) \\
& & P_3(r_1,r_2)   &\Leftarrow  r_1?\pa{n}. P_4(n,r_2)\\
& & P_4(n,r_2)   &\Leftarrow   r_2?\pa{d}. P_5(n,d) \\
& & P_5(n,d)     &\Leftarrow  \publish?\pa{z}. P_6(n,d,z)\\
& & P_6(n,d,z)   &\Leftarrow    z!\pc{n,d}. \CnewNT{r_1} P_1(r_1)\\
\\
\text{Advertiser:}& &\qquad
A_1      &\Leftarrow   \adv?\pa{r}. \CnewNT{d} A_{2}(r,d)  \\
& & A_2      &\Leftarrow   r!\pc{d}.A_1 
\\
% \text{News provider:} & &\qquad
% N_1            &\Leftarrow \news?\pa{r}. \CnewNT{n} N_2(r,n)\\
% & & N_2(r,n)  &\Leftarrow r!\pc{n}.N_1\\
\\
\text{Publisher with kickback:}\qquad &  &
P_{K1}(r_1)  &\Leftarrow   \news!\pc{r_1} \CnewNT{r_2,k} P_{K2}(r_1,r_2,k)\\
& & P_{K2}(r_1,r_2,k) &\Leftarrow \adv!\pc{k,r_2}. P_{K3}(r_1,r_2,k) \\
& & P_{K3}(r_1,r_2,k)   &\Leftarrow  r_1?\pa{n}. P_{K4}(n,r_2,k)\\
& & P_{K4}(n,r_2,k)   &\Leftarrow   r_2?\pa{d}. P_{K5}(n,d,k) \\
& & P_{K5}(n,d,k)     &\Leftarrow  \publish?\pa{z}. P_{K6}(n,d,k,z)\\
& & P_{K6}(n,d,k,z)   &\Leftarrow    k?. P_{K7}(n,d,z)\\
& & P_{K7}(n,d,z)   &\Leftarrow    z!\pc{n,d}. \CnewNT{r_1} P_{K1}(r_1)\\
\\
\text{Advertiser with kickback:} & & \qquad
A_{K1}      &\Leftarrow   \adv?\pa{k,r}. \CnewNT{d} A_{K2}(k,r,d)  \\
& & A_{K2}(k,r,d)      &\Leftarrow   r!\pc{d}.(A_{K1} \Cpar k!)\\
\end{alignat*}

  \caption{Notation for publisher code}
  \label{fig:pubcode}
\end{figure}

\begin{figure}[t]
  \begin{alignat*}{2}
\text{Standard publisher:} & & \qquad
    \text{PA}_{1} &\Leftarrow \CnewNT{\adv,r_1}(\Cloc{P_{1}(r_1)}{p} \;\Cpar\;\Cloc{A_{1}}{a})\\
&&    \text{PA}_{2}(r_1) &\Leftarrow \CnewNT{\adv,r_2}(\Cloc{P_{2}(r_1,r_2)}{p} \;\Cpar\;\Cloc{A_{1}}{a})\\
&&    \text{PA}_{3}(r_1) &\Leftarrow \CnewNT{\adv,r_2,d}(\Cloc{P_{3}(r_1,r_2)}{p} \;\Cpar\;\Cloc{A_{2}(r_2,d)}{a})\\
&&    \text{PA}_{4}(n) &\Leftarrow \CnewNT{\adv,r_2,d}(\Cloc{P_{4}(n,r_2)}{p} \;\Cpar\;\Cloc{A_{2}(r_2,d)}{a})\\
&&    \text{PA}_{5}(n) &\Leftarrow \CnewNT{\adv,d}(\Cloc{P_{5}(n,d)}{p} \;\Cpar\;\Cloc{A_{1} }{a})\\
&&    \text{PA}_{6}(n) &\Leftarrow \CnewNT{\adv,d}(\Cloc{P_{6}(n,d,r)}{p} \;\Cpar\;\Cloc{A_{K1}}{a})\\
\\\\
\text{Publisher with kickback:} & &\qquad
    \text{PA}_{K1} &\Leftarrow \CnewNT{\adv,r_1}(\Cloc{P_{K1}(r_1)}{p} \;\Cpar\;\Cloc{A_{K1}}{a})\\
&&    \text{PA}_{K2}(r_1) &\Leftarrow \CnewNT{\adv,r_2,k}(\Cloc{P_{K2}(r_1,r_2,k)}{p} \;\Cpar\;\Cloc{A_{K1}}{a})\\
&&    \text{PA}_{K3}(r_1) &\Leftarrow \CnewNT{\adv,k,r_2,d}(\Cloc{P_{K3}(r_1,r_2,k)}{p} \;\Cpar\;\Cloc{A_{K2}(k,r_2,d)}{a})\\
&&    \text{PA}_{K4}(n) &\Leftarrow \CnewNT{\adv,k,r_2,d}(\Cloc{P_{K4}(n,r_2,k)}{p} \;\Cpar\;\Cloc{A_{K2}(k,r_2,d)}{a})\\
&&    \text{PA}_{K5}(n) &\Leftarrow \CnewNT{\adv,k,r_2,d}(\Cloc{P_{K5}(n,d,k)}{p} \;\Cpar\;\Cloc{A_{K1} \Cpar k!}{a})\\
&&    \text{PA}_{K6}(n) &\Leftarrow \CnewNT{\adv,k,r_2,d}(\Cloc{P_{K6}(n,d,r,k)}{p} \;\Cpar\;\Cloc{A_{K1} \Cpar k!}{a})\\
&&    \text{PA}_{K7}(n) &\Leftarrow \CnewNT{\adv,d}(\Cloc{P_{K7}(n,d,r)}{p} \;\Cpar\;\Cloc{A_{K1} }{a})\\\\
  \end{alignat*}
  \caption{Publishing  systems}
  \label{fig:pubconfs}
\end{figure}

Here we revisit the publishing example developed in Example~\ref{ex:publishing}, Example~\ref{ex:kickback} and
Example~\ref{ex:publishingagain}; by exhibiting a witness bisimulation, again up to structural equivalence and 
$\beta$-moves, we show that
\begin{align}\label{ex:pub.provedagain}
  (\Gamma_{327} \with  \text{PA}_K)  \Eamort{0} ( \Gamma_{327} \with \text{PA})
\end{align}
subject to minor constraints on $\Gamma$; these constraints allow $\Gamma^o(\ownfnt p)$ to be finite. 
The systems $\text{PA}$ and $\text{PA}_K$, in addition to cost-free communications,
\begin{enumerate}[$\bullet$]
\item use resource $\news$; in the definition of the cost environment from Example~\ref{ex:publishing} this is
recorded as a loss of 3, the cost of using $\news$. In the abstract wLTS we are using this loss is paid for by
the funds in $\Gamma_{327}^o(\ownfnt p)$, while it costs nothing to provide

\item provide resource $\publish$; in the cost environment this is recorded as a gain of $6$, namely the difference between
providing it $7$ and using it $1$. Also this gain is added to the funds of $\Gamma_{327}^o(\ownfnt p)$.
\end{enumerate}
There are also internal communications which have costs associated with them, namely the use and provision of $\adv$; again
this is recorded as a loss of 2 which must be funded by $\Gamma_{327}^o(\ownfnt p)$. 

In order to describe the witness bisimulation we use the code abbreviations in Figure~\ref{fig:pubcode} and the  
system definitions in Figure~\ref{fig:pubconfs}. All environments we use have the form 
$\Gammadyn, 
\Cassoc{\news}{\tR_n},
\Cassoc{\publish}{\tR_p}$, and in order to fund the advertising we assume $\Gamma^o(\ownfnt a) = \infty$. 
In the witness bisimulation 
$\setof{\calR^k}{k \in \nats}$ all $\calR^k$ are identical and this unique relation $\calR$ is characterised by the following
constraints:
\begin{align*}
  \Gammaobs \with PA_{K1} \;\calR\; \Deltaobs \with PA_1  &\qquad\qquad 5 \leq \Gamma^o(\ownfnt p),\; 5 \leq \Delta^o(\ownfnt p)\\
  \Gammaobs \with PA_{K2}(r) \;\calR\; \Deltaobs \with PA_2 (r)  &\qquad\qquad 2\leq \Gamma^o(\ownfnt p),\; 
                              2 \leq \Delta^o(\ownfnt p),\; r \in \Chan\\
  \Gammaobs \with PA_{K3}(r) \;\calR\; \Deltaobs \with PA_3(r)   &\qquad\qquad  r \in \Chan\\
  \Gammaobs \with PA_{Ki}(n) \;\calR\; \Deltaobs \with PA_i(n)   &\qquad\qquad 4 \leq i \leq 6,\; n \in \News\\
  \Gammaobs \with PA_{K7}(n) \;\calR\; \Deltaobs \with PA_6(n)  &\qquad\qquad  n \in \News
\end{align*}
Here we use $\News$ to denote some set of news stories.

It is straightforward to show that this is indeed a weak amortised bisimulation in the abstract wLTS relative to the single
external observer $\sobs$. Since $\Gamma_{327} \with \text{PA}_K$ $\beta$-reduces to $\Gamma_{327} \with \text{PA}_{K1}$
and $\Gamma_{327} \with \text{PA}$ $\beta$-reduces to $\Gamma_{327} \with \text{PA}_1$, and 
$\Gamma_{327} \with \text{PA}_{K1} \;\calR\; \Gamma_{327} \with \text{PA}_1$, the required (\ref{ex:pub.provedagain}) above
follows. 

\section*{Acknowledgments}
The author would like to thank the referees for their very useful comments. 
\bibliographystyle{alpha}

\bibliography{buysell}

%\newpage

\end{document}